\documentclass[final,3p,12pt]{elsarticle}

\usepackage{amssymb}
\usepackage{amsthm}
\usepackage{amssymb}
\usepackage{amsmath}
\usepackage{amsfonts}
\usepackage{amsthm}
\usepackage{graphics}
\usepackage{subcaption}

\usepackage{epsfig}
\usepackage{natbib}
\usepackage{enumerate}

\usepackage{url}

\usepackage{epstopdf}
 
\newtheorem{theorem}{Theorem}

\usepackage{lineno}

\journal{Mathematical Biosciences}

\usepackage[usenames]{color}

\newcommand{\ddt}[1]{\partial_t{#1}}
\newcommand{\dda}[1]{\partial_a{#1}}

\newcommand{\ddx}[1]{\partial_x{#1}}
\newcommand{\ddy}[1]{\partial_y{#1}}

\renewcommand{\epsilon}{\varepsilon}
\renewcommand{\theta}{\vartheta}
\newcommand{\wB}{w^B}
\newcommand{\WB}{W^B}
\newcommand{\RM}{\mathcal{R}_0^M}
\newcommand{\RW}{\mathcal{R}_0^W}
\newcommand{\RWA}{\widehat{\mathcal{R}}_0^W}

\begin{document}

\begin{frontmatter}

%% Title, authors and addresses

%% use the tnoteref command within \title for footnotes;
%% use the tnotetext command for theassociated footnote;
%% use the fnref command within \author or \address for footnotes;
%% use the fntext command for theassociated footnote;
%% use the corref command within \author for corresponding author footnotes;
%% use the cortext command for theassociated footnote;
%% use the ead command for the email address,
%% and the form \ead[url] for the home page:
%% \title{Title\tnoteref{label1}}
%% \tnotetext[label1]{}

\title{A mathematical view on head lice infestations}

%\fntext[myfootnote]{corresponding author}

\author{Noemi Castelletti}
\address{Institute of Radiation Medicine, Helmholtz Zentrum M\"unchen, Ingolst\"adter Landstra{\ss}e~1,
D-85764 Neuherberg, Germany}
\ead{noemi.castelletti@helmholtz-muenchen.de}

\author{Maria Vittoria Barbarossa}
\address{Institute of Applied Mathematics, Heidelberg University, Im Neuenheimer Feld 205, D-69120 Heidelberg, Germany}
\ead{barbarossa@uni-heidelberg.de}

%%%%%%%%%%%%%%%%%%%%%%%%%%%%%%%%%%%%%%%%%%%%%%%%%%%%%%%%%%%%%%%%%%%%%%%%%%%%%%%%%%%%%%%%%%%%%%%%%%%%%%%%%%%%%%%%%%%%%%%

\begin{abstract}
Commonly known as head lice, \textit{Pediculus humanus capitis} are human ectoparasites which cause infestations in children worldwide. Understanding the life cycle of head lice is an important step in knowing how to treat lice infestations, as the parasite behavior depends considerably on its age and gender. In this work we propose a mathematical model for head lice population dynamics in hosts who could be or not quarantined and treated. Considering a lice population structured by age and gender we formulate the model as a system of hyperbolic PDEs, which can be reduced to compartmental systems of delay or ordinary differential equations. Besides studying fundamental properties of the model, such as existence, uniqueness and nonnegativity of solutions, we show the existence of (in certain cases multiple) equilibria at which the infestation persists on the host's head. Aiming to assess the performance of treatments against head lice infestations, by mean of computer experiments and numerical simulations we investigate four possible treatment strategies. Our main results can be summarized as follows: (i) early detection is crucial for quick and efficient eradication of lice infestations; (ii) dimeticone-based products applied every 4 days effectively remove lice in at most three applications even in case of severe infestations and (iii) minimization of the reinfection risk, e.g. by mean of synchronized treatments in families/classrooms is recommended.
\end{abstract}

\begin{keyword}
%% keywords here, in the form: keyword \sep keyword
head lice \sep population dynamics \sep differential equations \sep structured populations \sep {stability analysis} \sep {delay equations} 

\MSC  	92D25   \sep  	34K20 \sep  	34K17 \sep 34C60
\end{keyword}

\end{frontmatter}

%%%%%%%%%%%%%%%%%%%%%%%%%%%%%%%%%%%%%%%%%%%%%%%%%%%%%%%%%%%%%%%%%%%%%%%%%%%%%%%%%%%%%%%%%%%%%%%%%%%%%%%%%%%%%%%%%%%%%%%%

%%%% Introduction
\section{Introduction}
\label{sec:Introduction}
\textit{Pediculus humanus capitis}, commonly known as head lice, are obligate exclusively human ectoparasites, source of annoying infestations in children worldwide~\cite{Cummings2018,Feldmeier2014}. The main head lice transmission route being close head-to-head contact~\cite{Speare2002,Meister2016}, pediculus capitis epidemics occur mostly in schools and kindergartens, independently on the country of origin, ethnic groups and socio-economic status of the host~\cite{Feldmeier2012}. 	

Wingless and up to 4 mm long head lice live on the human scalp, feeding 4-8 times a day by  sucking  blood from the host and  injecting  saliva  simultaneously~\cite{TakanoLee2003,Cummings2018}. The life cycle of the louse is structured into five stages. After mating, a breeding females can lay up to 6 eggs per day for 30 days~\cite{Cummings2018}, close to the scalp, where temperature and humidity are optimal for their growth. Lice eggs (also called nits) hatch  6 to 11 days after ovoposition  into  nymphs  that  molt  twice  over  the  next  8 to 10 days to  become  adult lice. In contrast to nymphs after the first and second molt, first stage nymphs are not motile~\cite{TakanoLee2003}. Differentiation into female or male louse occurs after the third molt, when the insect becomes adult and sexually active. During mating both female and male louse might die~\cite{Liceworld}. Adult insects live about 30 days~\cite{TakanoLee2003}, but can survive for only 1 to 2 days away from the human host~\cite{Burkhart2003}.\\
\ \\
Pediculosis-induced itching occurs when the host develops an allergic reaction to the lice saliva, usually four to six weeks from the beginning of the infestation. Diagnosis of head lice infestations is based on the detection of adults,
nymphs and/or viable eggs on the host hair and scalp. Well-established treatment options for head lice infestations include therapeutic wet combing, topical application of a pediculicide, and oral treatments~\cite{Feldmeier2012}, the last not being considered in our study.
Wet combing is a non-chemical treatment, optimal for detection of head lice infestations~\cite{Gallardo2013}, but very time consuming when performed for taking advantage of its therapeutic effects~\cite{Feldmeier2012}. Most pediculicides, such as those based on malathion, pyrethrins and its synthetic derivates, can kill nymphs and adult lice, but are in general non-ovocidal. Extensive use of these compounds has led to the development of resistant head lice populations~\cite{Feldmeier2012,Cummings2018}. In contrast, dimeticone-based pediculicides showed moderate to high efficacy against live lice and eggs and the development of resistance to such products is less likely~\cite{Feldmeier2012}.

Whereas mathematical models for the dynamics and control of vector-borne diseases, such as mosquitoes or ticks, are well established (see e.g. Ch.4 in \cite{Martcheva2015} for an introductory overview), to the best of our knowledge only two groups have previously proposed mathematical approaches for understanding the spread of pediculosis. An epidemic model for transmission among hosts based on a stochastic SIS approach was suggested by Stone et al.~\cite{Stone2008}. This model describes only the macroscopic level (host interactions) and it does not consider the biology of the lice life cycle. Laguna and Risau-Gusman~\cite{Laguna2011} proposed a discrete model based on Leslie-Lefkovitch matrices and studied growth and interactions of colonies of head lice. This study was used in computer simulations to estimate the efficacy of different control strategies on the growth of the lice colony. In a recent work of the same authors~\cite{Laguna2018} the mathematical model was combined with data collected from schools in order to estimate the most likely events that can give rise to infestations.\\
\ \\
We propose here a mechanistic mathematical model for understanding the biology of the life cycle of head lice and assessing the efficiency of different treatments to eradicate lice infestations. Our first and more general approach is based on structured populations which are continuous in time and age, hence hyperbolic partial differential equations (PDEs). In contrast to the model by Laguna and Risau-Gusman~\cite{Laguna2011} we explicitly include the dynamics of the male lice and propose a mating function for pair formation. In Sect.~\ref{sec:model} we show under which conditions our PDE model can be reduced to systems of delay differential equations or ordinary differential equations. The latter are first analyzed (Sect.~\ref{sec:analysis}) and then used for computer experiments and numerical simulations (Sect.~\ref{sec:treat_simulax}) to investigate the efficacy of four possible treatments against head lice.

%%%% modeling
\section{Modeling head lice life cycle and transmission}
\label{sec:model}
In this section we propose mathematical models for head lice infestations based on the biology of the lice life cycle.  We first consider a lice population in an isolated environment, such as the head of an infected quarantined host. In a second step we extend the models to include lice transmission between hosts.

\subsection{Populations structured by age}
\label{sec:PDEmodel}
One possibility for modeling the lice life cycle is to use continuous age structures~\cite{Cushing1998}. Hoppensteadt~\cite{Hoppensteadt1975} introduced a mathematical model for a population structured by age with distinction of the two sexes and pair formation. We shall adapt and extend this approach.\\
\ \\
\noindent Let $w(t,a)$ denote the density of single female lice of age $a\geq0$ at time $t\geq 0$, that is, females which are not breeding and are available for mating. Respectively, we denote by $\wB(t,a)$ the density of breeding females, and by $m(t,a)$ the density of single male lice. We assume that the death rates are age-dependent functions. Birth rates are not relevant for the moment and shall be introduced later. Let $\mu_w(a)$ and $\mu_m(a)$ denote the death rate of female and male lice, respectively. In contrast to the model proposed in~\cite{Laguna2011}, in which it is assumed that males are readily available and that with only one fertilization
female lice are able to lay eggs until they die, we explicitly introduce the mating component and the possibility of multiple fertilization.
Let $p(t,x,y)$ be the lice pairs which, at time $t\geq 0$, are formed by females of age $x\geq 0$ and males of age $y\geq 0$. Pair formation is described by a function 
 $\phi(w,m)(t,x,y):= \tilde \rho(x,y)\pi(w(t,x),m(t,y))$, where $\tilde \rho(x,y)$ is the age-dependent mating rate and $\pi(w,m)$ describes the mating behavior. We assume that the following properties hold:
 \begin{enumerate}[(i)]
 	\item $\tilde \rho(x,y)\geq 0\,$ for all $x\geq0,\,y\geq0$,
 	\item $\pi(w,m)\geq 0\,$ and continuously differentiable in both variables. 
 \end{enumerate}
 As pair formation is not possible when only females or only males are present we require that
 \begin{enumerate}[(i)]
 	 \setcounter{enumi}{2}
 	\item $\pi(0,m)=0$, $\pi(w,0)=0$, for all $m\geq 0,\,w \geq 0$,
 	\item $\partial_w \pi(w,m)|_{(w,0)}=0 =\partial_m\pi(w,m)|_{(0,m)}$. 
 \end{enumerate}
Among the functions which satisfy the above assumptions~(i)-(iv), possible choices for $\pi(w,m)$ are given by the incidence law,
 \begin{equation*}
% 	\label{eq:matingfunction2}
 	\pi(w(t,x),m(t,y))= \frac{w(t,x)m(t,y)}{\int_0^{\infty} w(t,u)\,du+\int_0^{\infty}m(t,u)\,
 		du},
 \end{equation*}
 as suggested in \citep{Li2004} for mating of mosquitoes, or by the mass action law
 \begin{equation*}
% 	\label{eq:matingfunction}
 \pi(w(t,x),m(t,y))= w(t,x)m(t,y),
 \end{equation*}
often used for modeling contacts in epidemiological models~\cite{Martcheva2015}. In this paper we shall use the latter mating function. During a mating process both female and male louse might die. In particular it has been reported that if one of the two insects dies during the mating process, the other one dies as well~\cite{Liceworld}. To capture this phenomenon we introduce the probability $\xi\in [0,1]$ that a pair does not survive the mating process. Respectively, with probability $1-\xi$ both insects survive. We assume that pairs split at some constant rate $\sigma>0$, independent of the age of the insects. That is, for a pair formed by a female of age $x$ and a male of age $y$, let the pair splitting rate be $\sigma(x,y)\equiv \sigma\geq 0$. In time, single female lice age, might die due to natural death, and can mate with males.
It is still discussed whether a female which had a fertile mating will be breeding for its whole life~\cite{Maunder1993}, as it has been observed for the pubic louse \cite{Multiple_mating_Burgess}, or it needs a new mating for a new ovoposition \cite{Multiple_mating_Mehlhorn,Boutellis}. To keep the model as general as possible, we introduce the return rate, $\theta_{\alpha}$, of breeding females to the nonbreeding compartment. This parameter is defined as the product $\theta_{\alpha}:=\alpha \theta$, 
 where $1/\alpha>0$ corresponds to the average length of the breeding period and $\theta\in [0,1]$ is the probability that after the breeding period a female lice returns to the nonbreeding compartment. From balance laws and classical approaches for age-structured populations \cite{Webb2008} we obtain the equation
\begin{equation}
\label{model1_eq:PDEfemales}
\begin{aligned}
\ddt{w}(t,a)& = - \underbrace{\dda{w}(t,a)}_{\mbox{aging}}\,  -\,\underbrace{\mu_w(a)w(t,a)}_{\mbox{death}}\\
& \phantom{=}\; - \underbrace{\int_0^{\infty}\phi(w,m)(t,a,y)\,dy}_{\mbox{pair formation}}\; +
\underbrace{\theta_{\alpha}\wB(t,a).}_{\substack{\mbox{females which are}\\ \mbox{no longer  breeding}}}
\end{aligned}
\end{equation}
Analogously, the dynamics of the male population is given by  
\begin{equation}
\label{model1_eq:PDEmales}
\begin{aligned}
\ddt{m}(t,a)& = - \underbrace{\dda{m}(t,a)}_{\mbox{aging}}\,  -\,\underbrace{\mu_m(a)m(t,a)}_{\mbox{death}}\\
& \phantom{=}\; - \underbrace{\int_0^{\infty}\phi(w,m)(t,x,a)\,dx}_{\mbox{pair formation}}\; +\underbrace{\sigma (1-\xi)\int_0^{\infty}p(t,x,a)\,dx}_{\mbox{pairs split}}.
\end{aligned}
\end{equation}			
The equation for pairs is given by
\begin{equation}
\label{model1_eq:PDEpairs}
\begin{aligned}
\ddt{p}(t,x,y)& = -\ddx{p}(t,x,y)- \ddy{p}(t,x,y)-\underbrace{(1-\xi)\sigma p(t,x,y)}_{\mbox{pairs split}} \\[0.2em]
& \phantom{=}\;-\underbrace{\xi\sigma p(t,x,y)}_{\mbox{pairs die}} \;+ \;
\underbrace{\phi(w,m)(t,x,y)}_{\mbox{pair formation}}.
\end{aligned}
\end{equation}
After pair splitting the female moves to the breeding stage which culminates with an ovoposition. The dynamics of breeding females is given by 
\begin{equation}
\label{model1_eq:PDEfemalesB}
\begin{aligned}
\ddt{\wB}(t,a)& = -\dda{\wB}(t,a) -\mu_{\wB}(a)\wB(t,a)
-\theta_{\alpha} \wB(t,a)\\
& \phantom{=}\;  + \int_0^{\infty}(1-\xi)\sigma p(t,a,y)\,dy.
\end{aligned}
\end{equation}
New individuals are born by females in the breeding stage (in contrast, in \cite{Hadeler1993} they were born by pairs). Let $b(t,a)$ be the fertility rate of a breeding female of age $a$ at time $t$. With a certain probability $r\in [0,1]$ the egg will evolve into a male, respectively, with probability $1-r$ into a female. It is biological plausible to assume that there is no breeding female of age zero, nor pair in which one of the two insects is of age zero. Hence for all $t\geq 0$ we have the boundary conditions:
\begin{equation}
\begin{aligned}
w(t,0) &=(1-r)\int_0^{\infty}b(t,a)\wB(t,a)\,da,\\
m(t,0) &=r\int_0^{\infty}b(t,a)\wB(t,a)\,da,\\
\wB(t,0)&= 0,\;\mbox{and}\,\;
p(t,0,y) = 0\; = \;p(t,x,0) \quad \mbox{for all } x,\,y \geq 0.
\label{eq:initcond_pde}
\end{aligned}
\end{equation}
Nonnegative initial age distributions complete the model~\eqref{model1_eq:PDEfemales}--\eqref{eq:initcond_pde}.

\subsection{Transmission}
\label{sec:PDEwithTransmission}
Parasite transmission is dependent on the life cycle of the louse \citep{Meister2016}, in particular adult lice can move from host to host, while eggs or early stage nymphs are not motile~\citep{Liceworld}. We define the \textit{transferring rates} $\alpha_{w}(t,a)$ and $\beta_{w}(t,a)$ of female lice moving onto the host's head, respectively away from the host's head. Analogously, let $\alpha_{m}(t,a)$ and $\beta_{m}(t,a)$ be the transferring rates for males. Following~\citep{Liceworld}, we assume that neither breeding females nor pairs move from host to host. Then the equations \eqref{model1_eq:PDEfemales} and \eqref{model1_eq:PDEmales} change into 
\begin{equation}
\label{model1_eq:PDEfemales2}
\begin{aligned}
\ddt{w}(t,a)& = - \dda{w}(t,a)  -\mu_w(a)w(t,a) - \int_0^{\infty}\phi(w,m)(t,a,y)\,dy\\
& \phantom{=}\;  + \theta_{\alpha}\wB(t,a)
-\beta_w(t,a)w(t,a)+\alpha_w(t,a),
\end{aligned}
\end{equation}
respectively,
\begin{equation}
\label{model1_eq:PDEmales2}
\begin{aligned}
\ddt{m}(t,a) &  =-\dda{m}(t,a) -\mu_m(a)m(t,a)- \int_0^{\infty}\phi(w,m)(t,x,a)\,dx\\
& \phantom{=}\;+ \sigma (1-\xi)\int_0^{\infty}p(t,x,a)\,dx-\beta_m(t,a)m(t,a)+\alpha_m(t,a).
\end{aligned}
\end{equation}
Note that in equations~\eqref{model1_eq:PDEfemales2} and \eqref{model1_eq:PDEmales2} we chose age- and time-dependent transferring rates: dependence on age is for considering different transmission rates at different life stages of the louse, whereas dependence on time allows to model situations such as quarantine or interactions with other hosts (see Sec.~\ref{sec:treat_simulax}).

\subsection{From the age structure to delay equations}
\label{sec:ddemodel}
In spite of their elegance, continuous age-structured models such as~\eqref{model1_eq:PDEfemales}--\eqref{eq:initcond_pde} are hardly comparable to biological data from lice cultures experiments~{\cite{TakanoLee2003}}. Collected data are commonly of discrete nature, quantifying number of lice in a certain age group (or life stage). To provide a qualitative description of the head lice life cycle, such that it could be compared to experimental data, one might use compartmental models formulated as systems of ordinary differential equations (ODEs) or delay differential equations (DDEs). 
In the following we apply methods from \citep{Bocharov2000,MVBKPH2014} and show how to obtain a system of DDEs from the above PDE model \eqref{model1_eq:PDEfemales}--\eqref{eq:initcond_pde}. Let us suppose that we want to make use of mathematical models to understand lice reproduction or to fine-tune specific treatments which target adult lice only (or eggs only). Then we can simplify the continuous age structure in model \eqref{model1_eq:PDEfemales}--\eqref{eq:initcond_pde} and consider two classes of insects, namely, head lice in the juvenile phase (eggs and nymphs, $a\leq \tau$) and adult lice ($a>\tau$). During the juvenile phase of length $\tau$ days, lice are either in the egg stage or in one of the nymphs stages, and do not mate nor move. The biology suggests that $\tau \in [13-16]$ days ({cf. Sec.~\ref{sec:Introduction} and Table~\ref{Table_rates}}).  Let us define for all $t\geq0$ the following model variables:
\begin{align*}
J(t) & = \int_0^{\tau}(m(t,a)+w(t,a))\,da, \quad \mbox{the total number of juveniles},\\
W(t) & =\int_{\tau}^{\infty}w(t,a)\,da,\quad \mbox{the total number of nonbreeding adult females},\\
M(t)& =\int_{\tau}^{\infty}m(t,a)\,da,\quad \mbox{the total number of adult males},\end{align*}
\begin{align*}
P(t)& =\int_{0}^{\infty}\int_{0}^{\infty}p(t,x,y)\,dx\,dy,\quad \mbox{the total number of pairs},\\
\WB(t)& =\int_{\tau}^{\infty}\wB(t,a)\,da,\quad\mbox{the total number of adult breeding females}.
\end{align*}
We characterize these populations in terms of fertility, death, motility and mating rates. We assume that single females and male lice die at the same rate, hence:
\begin{displaymath}
\mu_w(a)\equiv \mu_m(a) = \begin{cases}
\mu_0 & a\leq \tau,\\
\mu_1 & a> \tau,
\end{cases}
\quad \mbox{with}\quad \mu_1\geq \mu_0>0.
\end{displaymath}
Pair formation occurs only among adult lice, thus for the age-dependent mating rate $\tilde \rho(x,y)$ we set
\begin{displaymath}
\tilde \rho(x,y)=\begin{cases}
0 & \mbox{if}\quad x\leq \tau \quad\mbox{or}\quad y\leq \tau,\\
\rho & \mbox{if}\quad x> \tau \quad\mbox{and}\quad y>\tau.
\end{cases}
\end{displaymath}
As a result, $p(t,x,y)=0$ for $x\leq\tau$ or $y\leq\tau$. It follows that there is no breeding female of age $x\leq\tau$, that is,
\begin{equation}
\label{model1_eq:pop_t_tau_0}
\begin{aligned}
p(t,x,y) &= 0,\quad \mbox{for} \; x \leq \tau \; \mbox{or}\; y \leq \tau,\\
\wB(t,a) &= 0,\quad  a \leq \tau.
\end{aligned}
\end{equation}
For simplicity, let us assume that the fertility rate of breeding females,
\begin{displaymath}
b(t,a)=\begin{cases}
0 & \mbox{if}\quad a\leq \tau,\\
b_1\geq 0 & \mbox{if}\quad a> \tau,
\end{cases}
\end{displaymath}
and the death rate, $\mu_{\wB}(a)=\mu_B>0$, are constant values for all $t\geq 0,\, a>\tau$. It is biologically meaningful to assume that there is no "infinitely old" female, that is, $w(t,a)\to 0$ for $a\to \infty$. Similarly,
\begin{equation*}
%\label{model1_eq:pop_a_to_infty}
\begin{aligned}
m(t,a) &\to 0,\quad \wB(t,a)\to 0\quad \mbox{for} \quad a\to \infty,\\
p(t,x,y)&\to 0,\quad \mbox{for}  \quad x\to \infty \quad \mbox{or} \quad y\to
\infty.
\end{aligned}
\end{equation*}
As in Sect.~\ref{sec:PDEwithTransmission}, we use transmission coefficients to observe head-lice moving from one head to another. Under the assumption that juvenile lice do not move, we set 
$$\beta_w(t,a)= \begin{cases}
0 & a\leq \tau,\\
\beta_W(t)\geq 0 & a> \tau,
\end{cases}\qquad 
\alpha_w(t,a)= \begin{cases}
0 & a\leq \tau,\\
\alpha_W(t)\geq 0 & a> \tau,
\end{cases}.
$$
Analogously, we set $\beta_m(t,a)=\beta_M(t) \geq 0$, respectively $\alpha_m(t,a)=\alpha_M(t)\geq 0$, for $a>\tau$ and zero otherwise.
Under the above assumptions, differential equations for variables $J,\,M,\,W$ and $\WB$ can be rigorously obtained (cf.~\citep{MVBKPH2014,Bocharov2000}) from the system \eqref{model1_eq:PDEfemales}--\eqref{eq:initcond_pde}. We show in the following how to obtain the equation for the juveniles.
\begin{align*}
\dot J(t)  & = \ddt{\,}\int_0^{\tau}\bigl(w(t,a)+m(t,a)\bigr)\,da\\
&  = -\int_0^{\tau}\bigl(\dda{w}(t,a)+\dda{m}(t,a)\bigr)\,da\\
& \phantom{=} -\int_0^{\tau}\bigl(\mu_w(a)w(t,a)+\mu_m(a)m(t,a)\bigr)\,da\\
& \phantom{=} -\int_0^{\tau}\beta_m(t,a)m(t,a)\,da+\int_0^{\tau}\alpha_m(t,a)\,da\\
& \phantom{=} -\int_0^{\tau}\beta_w(t,a)w(t,a)\,da+\int_0^{\tau}\alpha_w(t,a)\,da\\
& = -w(t,\tau)-m(t,\tau)+w(t,0)+m(t,0)-\mu_0J(t)\\
& \underset{\eqref{eq:initcond_pde}}{=}	-w(t,\tau)-m(t,\tau)+\int_0^{\infty}b(t,a)\wB(t,a)\,da-\mu_0J(t)\\	
& =	-w(t,\tau)-m(t,\tau)+\int_{\tau}^{\infty}b(t,a)\wB(t,a)\,da-\mu_0J(t)\\	
& =	-w(t,\tau)-m(t,\tau)+b_1\WB(t)-\mu_0J(t).	
\end{align*}
In the last expression we still find the addends $w(t,\tau),\,m(t,\tau)$, related to the PDE approach. These shall be formulated in terms of the new variables $J,\,M,\,W,\,\WB$. Applying the method of characteristics, for $t>\tau$ we find
\begin{equation*}
%\label{model1_eq:fem_males_t_tau}
\begin{aligned}
w(t,\tau) &= w(t-\tau,0)e^{-\mu_0\tau}\,=\,(1-r)b_1\WB(t-\tau)e^{-\mu_0\tau},\\
m(t,\tau) &= m(t-\tau,0)e^{-\mu_0\tau}\,=\,rb_1\WB(t-\tau)e^{-\mu_0\tau}.
\end{aligned}
\end{equation*}
Hence, the equation for the juveniles is given by
\begin{equation}
\label{model1_eq:juv}
\dot J(t) = b_1\WB(t)-b_1\WB(t-\tau)e^{-\mu_0\tau}-\mu_0J(t).	
\end{equation}
Similarly, one can obtain the equations for adult females, adult males and breeding females. For the total number of pairs it is useful to recall the condition \eqref{model1_eq:pop_t_tau_0}. 

Then we have
\begin{align*}
\dot P(t) 
& = \ddt{\,}\int_{0}^{\infty}\int_{0}^{\infty}p(t,x,y)\,dx\,dy\\
& = -\int_{\tau}^{\infty}\int_{\tau}^{\infty}\ddx{p}(t,x,y)\,dx\,dy - \int_{\tau}^{\infty}\int_{\tau}^{\infty}\ddy{p}(t,x,y)\,dx\,dx\\
&\phantom{=} -\int_{\tau}^{\infty}\int_{\tau}^{\infty}\sigma p(t,x,y)\,dx\,dy+\int_{\tau}^{\infty}\int_{\tau}^{\infty}\phi(w,m)(t,x,y)\,dx\,dy\\
& = -\sigma P(t) +\rho W(t)M(t).
\end{align*}
{Mating is rather fast compared to other processes, such as death or reproduction, in the life cycle of head lice~\cite{Liceworld}. Hence, we can assume the pairs dynamics to occur on a fast time scale, hence that it holds $\epsilon \dot P= -\sigma P +\rho W M$, for $\epsilon >0$ small. Considering the limit $\epsilon \to 0$ we obtain the quasi-steady state approximation, $P=\rho WM/\sigma$, and substitute this in the equations for $M$ and $\WB$.} Thus the DDE model is reduced to a system of four equations:
\begin{equation}
\label{model1_eq:DDE_all_inout}
\begin{aligned}
\dot J(t) & =b_1\WB(t)-b_1\WB(t-\tau)e^{-\mu_0\tau}-\mu_0J(t)\\
\dot W(t)  & =(1-r)b_1\WB(t-\tau)e^{-\mu_0\tau} -(\mu_1+\rho M(t)-\beta_W(t))W(t)\\ 
& \phantom{=} \; +\theta_{\alpha}\WB(t)+\alpha_W(t)\\
\dot M(t) & = rb_1\WB(t-\tau)e^{-\mu_0\tau}-(\mu_1+\xi W(t)-\beta_M(t))M(t)+\alpha_M(t)\\
\dot\WB(t)  &= -(\theta_{\alpha}+\mu_{B})\WB(t) +(1-\xi)\rho W(t)M(t).
\end{aligned}
\end{equation}
A similar model was proposed in the master thesis of the first author \citep{Castelletti2015}. {It can be observed that whereas in the equations for W and M the delay  appears in form of a positive feedback term, in the juvenile population we find a negative feedback due to maturation $(-b_1\WB(t-\tau)e^{-\mu_0\tau})$. In the unfortunate case $J(0)=0,\, \WB(0)=0$ and $\WB(t)>0$ for $t<0$ this would lead to a negative solution for $J$. For guaranteeing nonnegativity of solutions proper initial data can be chosen, as we explain below.}\\
\ \\
A general expression for the number of juvenile lice at time $t\geq 0$ is given by 
\begin{equation}
J(t) = \int_0^t b_1 \WB(v) e^{-\mu_0(t-v)} \Pi(t-v)\,dv,
\label{eq:egg_integ}
\end{equation}
meaning that juveniles at time $t$ are those eggs deposited in the interval of time $[0,t]$, which did not die nor exited the juvenile compartment. The probability $\Pi(a)$ relates to the maturation rate and transition to the adult compartment of juveniles of age $a$. Dependent on the choice of the probability distribution one can obtain various types of differential equations~\cite{Yuan2014}. For example, if we choose the uniform distribution  
$$\Pi(a)= \begin{cases}
1, \qquad \mbox{for}\; a \in [0,\tau]\\
0, \qquad \mbox{otherwise},
\end{cases}$$
then substitution in \eqref{eq:egg_integ} yields
\begin{equation}
\label{eq:UintegralDDE}
J(t) = \int_0^{\tau} b_1 \WB(t-v) e^{-\mu_0 v}\,dv = \int_{t-\tau}^{t} b_1 \WB(w) e^{-\mu_0(t-w)}\,dw. 
\end{equation}
Differentiating~\eqref{eq:UintegralDDE} with respect to the time $t$ we obtain equation \eqref{model1_eq:juv}. Equation~\eqref{eq:UintegralDDE} could be used for defining biologically meaningful initial data for the DDE model~\eqref{model1_eq:DDE_all_inout} and guaranteeing nonnegativity of solutions.

\subsection{When treatments target specific life stages}
\label{sec:ODEmodel}
Treatments of head lice infestations are based on products (cf. Sec.~\ref{sec:Introduction}) which target specific life stages of the lice. In particulars certain treatments are effective only against nymphs and adult lice, whereas others target also eggs. To model treatments we shall use time dependent functions. Further we shall include compartments for all lice stages of interest, in particular we shall refine the juvenile class in~\eqref{eq:egg_integ} and consider eggs and nymphs separately.\\
\ \\
We go back to the integral equation~\eqref{eq:egg_integ}. Choosing for $\Pi$ an exponential distribution,
\begin{equation*}
\Pi(a)=e^{-\eta a},
%\label{expodistri}
\end{equation*}
and substituting in \eqref{eq:egg_integ}, we obtain
\begin{equation*}
J(t) = \int_0^{t} b_1 \WB(t-v) e^{-\mu_0 (t-v)}e^{-\eta (t-v)}\,dv.
\end{equation*}
Differentiation with respect to the time $t$ yields the ordinary differential equation
\begin{equation*}
\dot J(t) = b_1 \WB(t)-(\mu_0 +\eta)J(t).
%\label{eq:juv_ode}
\end{equation*}
\noindent The maturation term $\eta J(t)$ indicates that $1/\eta$ is the average duration of the juvenile stage. Assuming as above that transitions between one life stage and the next follow an exponential distribution, one can introduce a new compartment for each life stage to observe. The result is a system of ODEs with linear transitions between compartments~\citep{MacDonald1978}, namely
\begin{equation}
\begin{aligned}
\dot{U}(t) & = b_1 W^B(t) -(\mu_0 +\eta) U(t) -T_U(t) U(t)\\[0.5em]
\dot{N}(t) & = \eta U(t) -(\omega +\mu_N +T(t))N(t)\\[0.5em]
\dot{W}(t) & = (1-r)\omega N(t) -(\mu_1 +\rho M(t) +T(t)+\beta_W(t))W(t) \\
&\qquad +\theta_{\alpha} W^B(t)+\alpha_W(t)\\[0.5em]
\dot{M}(t) &= r\omega N(t)-(\mu_1 +\xi W(t)+\beta_M(t)+T(t))M(t)+\alpha_M(t)\\[0.5em]
\dot{W^B}(t) &= (1-\xi)\rho M(t)W(t) -(\mu_B +\theta_{\alpha}+T(t))W^B(t),
\end{aligned}
\label{ModelODE}
\end{equation}
with sub-populations for eggs ($U$), nymphs ($N$), single adult females ($W$), males ($M$) and breeding females ($W^B$). Maturation from eggs to nymphs, and from nymphs to adult, occurs at rate $\eta>0$ and $\omega > 0$, respectively. A schematic representation of the mathematical model \eqref{ModelODE} is presented in Fig.~\ref{Fig1_CompModel}. The terms $T_U(t)$ and $T(t)$ describe the effect of treatments against eggs and nymphs/adult lice, such as shampoos or combs. For the numerical tests shown in Sect.~\ref{sec:treat_simulax} we will choose these functions to be equal to zero in absence of treatment and nonzero at the time of the treatment. 
\begin{figure}[t]
	\centering
	\includegraphics[width=0.9\linewidth]{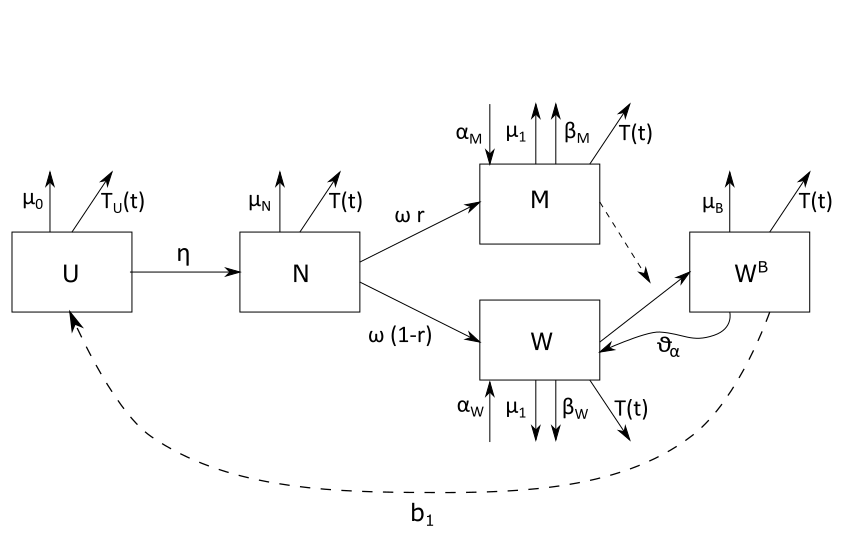}
	\caption{A schematic representation of the five compartment model \eqref{ModelODE}, with sub-populations for eggs ($U$), nymphs ($N$), females ($W$), males ($M$) and breeding females~{($W^B$).} The model includes reproduction ($b_1$), maturation from one compartment into the next one ($\eta,\, \omega$), death ($\mu_j$), migration ($\alpha_j, \beta_j$) and time-dependent treatments ($T_U(t),\,T(t)$).}
	\label{Fig1_CompModel}
\end{figure}
\subsection{Limit cases}
We conclude this section with few considerations on the limit cases,\linebreak $\theta=0,1$ and $\xi=0,1$. When $\xi=0$, no louse dies during the mating process. This assumption simplifies the equations for $M(t)$ and $\WB(t)$, as all females which mate will be able to lay eggs. In contrast, the case $\xi=1$ is not of biological interest. If all insects which mate die, the whole population will soon die out.\\
\indent Setting $\theta=0$, hence $\theta_{\alpha}=0$, one assumes that after the first mating a female will lay eggs for its whole lifetime, as it was assumed in~\cite{Laguna2011}. On the other hand, $\theta=1$, that is, $\theta_{\alpha}=\alpha$, means that females are breeding for a time $1/\alpha$, then need to mate again for a new ovoposition.

 \section{Analytical results}
 \label{sec:analysis}
In this section we present analytical results on the autonomous version of model~\eqref{ModelODE}, with $T(t)\equiv 0 \equiv T_U(t)$, for all $t\geq 0$. In the first step we consider isolated hosts, hence no lice transmission ($\alpha_j (t)\equiv 0 \equiv \beta_j(t)$, $j=W,M$), and study existence and stability of equilibria. 

\begin{theorem}
\label{exist_uniq_posit}
Consider system~\eqref{ModelODE} with $\alpha_j (t)\equiv 0 \equiv \beta_j(t)$, $j=W,M$, $T(t)\equiv 0 \equiv T_U(t)$ for all $t\geq 0$, and let nonnegative initial values be given. Then the autonomous system admits a unique nonnegative solution.
 \end{theorem}
\begin{proof}
Existence and uniqueness of the solution are guaranteed by the theorem of Picard-Lindel\"of and the proof is simple given the smoothness of the right-hand side. For the nonnegativity of solutions we consider the right-hand side at the boundaries of the positive cone. Let us assume that for some $\bar t>0$ the solution $U(\bar t)=0$, while all others components are nonnegative, then we have $\dot U(\bar t) = b_1 \WB(\bar t)\geq 0$. Hence the component $U$ does not drop below zero. Similarly, one can show nonnegativity for all other model components.
\end{proof}
 
\begin{theorem}
\label{prop:LFE}
Consider system~\eqref{ModelODE} with $\alpha_j (t)\equiv 0 \equiv \beta_j(t)$, $j=W,M$, and $T(t)\equiv 0 \equiv T_U(t)$ for all $t\geq 0$. There is only one lice-free-equilibrium (LFE) $P_0=(0,0,0,0,0)$, and there is no other equilibrium of the system in which at least one component is equal to zero.  The lice-free-equilibrium is locally asymptotically stable. 
  \end{theorem} 	
 \begin{proof}
 We omit the trivial computation to show that the LFE is unique and that there is no further equilibrium with one or more components equal to zero. Shortly, if any component of the equilibrium is zero, then recursively all other components turn out to be zero.\\
 \ \\
 For the proof of local asymptotic stability we linearize about $P_0$ and obtain the linear system $\dot Z(t) = J(P_0)Z(t)$, with the Jacobian matrix
\begin{displaymath}
J(P_0)=
 \begin{pmatrix}
 	-\mu_0-\eta & 0 & 0 & 0 & b_1\\
 	\eta & -\omega -\mu_N & 0 & 0 & 0\\
 	0 & (1-r)\omega & -\mu_1 & 0 & \theta_{\alpha}\\
 	0 & r\omega & 0 & -\mu_1 & 0\\
 	0 & 0 & 0 & 0 & -\mu_B -\theta_{\alpha}
 \end{pmatrix}.
\end{displaymath} 
 Local stability of $P_0$ is determined by the real parts of the roots $\lambda$ of the characteristic polynomial, $\det(J(P_0-\lambda I))=0$. 
 As by assumption all parameters in~\eqref{ModelODE} are nonnegative, and in particular all death rates are strictly positive, it can be quickly shown that $\lambda_j<0,\, j=1,\ldots,5$. Thus, the LFE is a locally asymptotically stable node.
 \end{proof}

Assume from here on that $\xi>0, \rho>0$. For the proof of existence and uniqueness of a nontrivial equilibrium it is convenient to define the nonnegative constants
\begin{equation*}
\begin{aligned}
\RW & := \frac{(1-\xi) r \omega \eta b_1 \rho }{\xi (\mu_0+\eta)(\omega+\mu_N)(\mu_{B}+\theta_{\alpha})},\\[0.8em]
\RM & := \frac{(1-\xi)\left[(1-r)\omega \eta b_1 +\theta_{\alpha}(\mu_0 +\eta)(\omega+\mu_N)\right]}{(\mu_0 +\eta)(\omega+\mu_N)(\mu_B +\theta_{\alpha})}.
\end{aligned}
\end{equation*}
These can be interpreted as "basic reproduction numbers" of adult male ($\RM$) and female ($\RW$) lice, respectively.
\begin{theorem}
\label{exi_uni_non_triv_equi}
Consider system~\eqref{ModelODE} with $\alpha_j (t)\equiv 0 \equiv \beta_j(t)$, $j=W,M$, and $T(t)\equiv 0 \equiv T_U(t)$ for all $t\geq 0$. If $\RM>1$ and $\RW>1$ then there is a unique positive equilibrium point of system~\eqref{ModelODE},  
$P_1=(U_1^*,N_1^*,W_1^*,M_1^*,W_1^{B*}),$ which is {unstable}. The coordinates of $P_1$ are given by
\begin{equation}
\begin{aligned}
W_1^* &=\frac{\mu_1}{\xi(\RW-1)}, \qquad
M_1^* = \frac{\mu_1}{\rho(\RM-1)}, \\ 
W^{B*}_1 &= \frac{\mu_1^2 (1-\xi)}{(\RM-1)(\RW-1)(\mu_{B}+\theta_{\alpha})},\\
U_1^*  &=\frac{b_1}{\mu_0+\eta}W^{B*}_1,\qquad
N_1^*  = \frac{\eta}{\omega+\mu_N}U^{*}_1.
\end{aligned}
\label{non_trivial_equil}
\end{equation}
\end{theorem}

\begin{proof}
The equilibrium conditions are obtained by setting the right-hand side of system~\eqref{ModelODE} equal to zero,
\begin{align}
0 =&\, b_1 {\WB}^* -(\mu_0 +\eta) U^* \label{eq:equiU1}\\
0 =&\, \eta U^* -(\omega +\mu_N)N^* \label{eq:equiN1}\\
0 =&\, (1-r)\omega N^* -(\mu_1 +\rho M^*)W^* +\theta_{\alpha}{\WB}^* \label{eq:equiW1}\\
0=& \,r\omega N^*-(\mu_1 +\xi W^*)M^* \label{eq:equiM1}\\
0=& \,(1-\xi)\rho M^*W^* -(\mu_B +\theta_{\alpha}){\WB}^*\label{eq:equiWB1}.
\end{align}
From the first and second equation~\eqref{eq:equiU1},\eqref{eq:equiN1} we calculate, respectively, $U^*$ and $N^*$ as (linear) functions of ${\WB}^{*}$,
\begin{equation}
U^*=\frac{b_1}{\mu_0 +\eta} {\WB}^{*}, \qquad 
N^*=\frac{\eta}{\omega+\mu_N} U^* = \frac{\eta}{\omega+\mu_N}\frac{b_1}{\mu_0 +\eta} {\WB}^{*}.
\label{eq:equiUN2}
\end{equation}
With the last relation we obtain an expression for $W^{B*}$ as a function of $M^*$ and $W^*$,
\begin{equation}
{\WB}^{*}= \frac{(1-\xi)\rho}{\mu_B +\theta_{\alpha}} M^* W^*. \label{eq:equiWB2}
\end{equation}
Now we substitute~\eqref{eq:equiUN2}, \eqref{eq:equiWB2} into the equation~\eqref{eq:equiW1} and find a linear equation in $M^*$ which (assuming $W^*\neq 0$) yields 
\begin{equation*}
M^* = \frac{\mu_1}{\frac{(1-\xi)(1-r)\omega \eta b_1 \rho}{(\mu_0 +\eta)(\omega+\mu_N)(\mu_B +\theta_{\alpha})} +\theta_{\alpha}\frac{(1-\xi)\rho}{\mu_B +\theta_{\alpha}} -\rho} = \frac{\mu_1}{\rho(\RM-1)}.
\end{equation*}
This value is nonnegative if $\RM>1$. Similarly, assuming  $M^*\neq 0$, from \eqref{eq:equiM1} we find 
   \begin{equation*}
W^*= {\frac{\mu_1}{\frac{(1-\xi) r\omega \eta b_1 \rho}{(\mu_0 +\eta)(\omega+\mu_N)(\mu_B +\theta_{\alpha})} -\xi}} = \frac{\mu_1}{\xi(\RW-1)},
\end{equation*}
hence $W^*>0$ if $\RW>1$.
If either $W^*$ or $M^*$ are zero, we are in the case of Theorem~\ref{prop:LFE} and we find the lice-free equilibrium $P_0$.\\
\ \\ 	
For the proof of linearized stability we introduce a slightly more compact notation and define the constants
\begin{equation}
\begin{aligned}
&\kappa_1= 1-\xi \in[0,1], \;\; \kappa_2=\theta_{\alpha}+\mu_B>0,  \;\; \kappa_3 = \mu_0 +\eta>0,  \\
&\kappa_4= (1-r)\omega>0,\;\; \kappa_5 = r\omega>0,  \;\; \kappa_6 = \omega +\mu_N>0,  \; \;
\kappa_7=\eta b_1>0. 
\end{aligned}
\label{kappas}
\end{equation}
We consider the linearized system about $P_1$, governed by the Jacobian matrix 
\begin{displaymath}
J(P_1)=
\begin{pmatrix}
-\kappa_3 & 0 & 0 & 0 & b_1\\
\eta & -\kappa_6 & 0 & 0 & 0\\
0 & \kappa_4 & -\mu_1 -\rho M^* & -\rho W^* & \theta_{\alpha}\\
0 & \kappa_5 & -\xi M^* & -\mu_1 -\xi W^* & 0\\
0 & 0 & \kappa_1\rho M^* & \kappa_1\rho W^* & -\kappa_2
\end{pmatrix}.
\end{displaymath} 
Long computation leads to the characteristic polynomial of $P_1$,
\begin{equation}
\begin{aligned}
f(\lambda)& = -f_a(\lambda)+f_b(\lambda)\\
& =
-\underbrace{(\kappa_3+\lambda)(\kappa_6+\lambda)(\mu_1+\lambda)(\lambda^2+\lambda \phi_1+\phi_0)}_{f_a(\lambda)} +\underbrace{\chi_0 +\chi_1 \lambda}_{f_b(\lambda)},
\label{poly1}
\end{aligned}
\end{equation}
where
\begin{align*}
	\phi_1 & := {\kappa_2}+\mu_1+\xi W^* +\rho M^*>0\\
	\phi_0 & := (\theta_{\alpha}\xi+\mu_{B})\rho M^* +\kappa_2(\mu_1+\xi W^*) >0 \\
	\chi_1 & = \kappa_1 \kappa_7 \rho (\kappa_4 M^* +\kappa_5  W^* )>0,\\
	\chi_0 & = \mu_1 \chi_1>0.
\end{align*}
The local asymptotic stability of $P_1$ is determined by the zeros of $f(\lambda)$ in \eqref{poly1}, or equivalently by the intersections of the fifth order curve $f_a(\lambda)$ with the line $f_b(\lambda)=\chi_0+\chi_1 \lambda$. The latter has positive slope and positive intercept with the y-axis. The quadratic factor in $f_a(\lambda)$ can be written as the product  $(\lambda-A_1)(\lambda-A_2)$, with
$A_{1,2}= (-\phi_1 \pm\sqrt{\phi_1^2-4 \phi_0})/2$. Observe that $A_{1,2}$ are either both real and negative (because $\phi_0,\phi_1>0$), or complex conjugated with negative real part. Hence $f_a(\lambda)$ is a fifth order polynomial with zeros laying all on the left half of the complex plane. Further, it has positive intercept with the y-axis, $f_a(0)=\kappa_3\kappa_6\mu_1 \phi_0>0$. Moreover, it holds that $f_b(0)>f_a(0)>0$. Indeed
\begin{align*}
\lefteqn{f_b(0)-f_a(0)}\\
& = \mu_1 \left( \kappa_1 \kappa_7 \rho (\kappa_4 M^* +\kappa_5 W^*) -\kappa_3 \kappa_6 \phi_0\right) \\
&= \mu_1 \left(\rho M^*\underbrace{\left(\rho \kappa_1\kappa_4 \kappa_7  -\kappa_3 \kappa_6 \rho(\kappa_2 -\theta_{\alpha}\kappa_1)\right)}_{=(\RM-1)\kappa_3 \kappa_6\kappa_2}\right. \\
& \phantom{=} \qquad \qquad \left. + W^*\underbrace{(\rho\kappa_1 \kappa_5\kappa_7  -\xi \kappa_2 \kappa_3 \kappa_6)}_{=\xi(\RW-1)\kappa_2\kappa_3 \kappa_6} -\mu_1\kappa_2\kappa_3 \kappa_6 \right) \\
&= \mu_1^2 \kappa_2 \kappa_3 \kappa_6 >0.
\end{align*}
It follows that the characteristic polynomial \eqref{poly1} has at least one root with positive real part, hence the coexistence equilibrium $P_1$ is unstable.
\end{proof}

\noindent Let us assume the infected host is not quarantined, that is, $\alpha_j(t),\,\beta_j(t)$, $j=M,W$ are not identically zero for $t\geq 0$. For simplicity of calculation,  we consider the special case of constant transferring rates, $ \alpha_j(t)= \hat \alpha_j>0$ and $\beta_j(t)=\hat \beta_j>0,\,j=M,W$.	

\begin{theorem}
Let $\alpha_j(t)=\hat \alpha_j>0$ and $\beta_j(t)=\hat \beta_j>0$, $t\geq 0$, $j=M,W$ and  $T(t)\equiv 0 \equiv T_U(t)$ for all $t\geq 0$ in system~\eqref{ModelODE}. Assume $\RM>1$, $\RW>1$ and let $M_1^*$ be the male component of $P_1$, as indicated in \eqref{non_trivial_equil}. Let further
{$\phi_b<0$,} $0<\Delta_{\phi}:=\phi_b^2-4\phi_a\phi_c$, and ${\RWA}>1$, where
\begin{equation}
\begin{aligned}
&\phi_a := \xi(\RW-1)\left(M_1^*+\frac{\beta_W}{\rho(\RM-1)}\right), \\
& \phi_b :=\alpha_M- \frac{\alpha_{W}\xi(\RW-1)+(\mu_1+\beta_{M})(\beta_{W}+M_1^*\rho(\RM-1))}{\rho(\RM-1)}, \\
& \phi_c := \alpha_W\frac{\mu_1+\beta_M}{\rho(\RM-1)},\qquad
{\RWA}:= \frac{\beta_{W}}{\alpha_{W}} \frac{\sqrt{\Delta_{\phi}}-\phi_b}{2\phi_a}.
\end{aligned}
\label{conditions2}
\end{equation}
Then system~\eqref{ModelODE} has two positive equilibria, $P_j=(U_j^{*},N_j^{*},W_j^{*},M_j^{*},W_j^{B*})$, $j=2,3$, where $W_{2,3}^* =(-\phi_b\pm\sqrt{\Delta_{\phi}})/2\phi_a$ and
\begin{align*}
M_j^* &= M_1^*+\frac{\beta_W}{\rho(\RM-1)}-\frac{\alpha_W}{\rho(\RM-1)W_j^*},\quad j=2,3\\ %\label{endemic_non_trivial_equil}\\ 
W^{B*}_j &= \frac{\rho(1-\xi)}{\theta_{\alpha}+\mu_B}M_j^*W_j^*,\quad
U_j^*  =\frac{b_1}{\mu_0+\eta}W^{B*}_j,\quad
N_j^*  = \frac{\eta}{\omega+\mu_N}U^{*}_j,\quad j=2,3. \nonumber
\end{align*}
\end{theorem}

\begin{proof}
Also in this proof we use the compact notation~\eqref{kappas}. We first show that when the transferring rates are nonzero, the LFE is not an equilibrium. As in the proof of Theorem 3 we denote by $X^*$ the steady state of the variable $X$. Let us assume $N^{*}\geq 0$. The steady state conditions yield

$$ M^{*}=0 \Leftrightarrow \underbrace{\kappa_5 N^{*}}_{\geq 0}= \underbrace{-\hat \alpha_M}_{<0},$$
which is a contradiction. On the other hand, if $W_B^*=0=N^*$ then from the $W$-equation we find $\hat \alpha_W=0$ which contradicts the assumption on the positive transferring rates. Hence, $M^{*}>0, W^{*}>0$, implying all other components are also nonzero.\\
\ \\
Now we compute the nontrivial equilibrium, analogously to $P_1$ in Theorem~\ref{exi_uni_non_triv_equi}. The relations~\eqref{eq:equiU1},  \eqref{eq:equiN1} and \eqref{eq:equiWB1} hold true also in the case of a non-quarantined host. We consider the algebraic equation given by $\dot W=0$. Assuming $W^* \neq 0$, we obtain an expression for $M^*$ as a function of $W^*$,
\begin{align}
M^* &= \underbrace{\frac{(\mu_1 +\beta_W)W^* -\alpha_W}{\rho W^*\left(\frac{\kappa_1\kappa_4\kappa_7}{\kappa_2\kappa_3\kappa_6} +\frac{\kappa_1}{\kappa_2}\theta_{\alpha} -1\right)}}_{=\rho (\RM-1)W^*} \nonumber \\
\phantom{M^* } &=M_1^* + \frac{\beta_W}{\rho(\RM -1)} -\frac{\alpha_W}{\rho(\RM -1)W^*}. \label{M_P23}
\end{align}
Observe that $M_1^*$ is nonnegative if $\RM>1$. Hence $\RWA>1$ provides a sufficient condition for $M^*>0$. From the algebraic equation $\dot M=0$ we calculate
\begin{equation*}
\underbrace{M^*}_{equ.~\eqref{M_P23}} \left( \underbrace{\left(\frac{\kappa_1\kappa_5\kappa_7}{\kappa_2\kappa_3\kappa_6}\rho -\xi \right)}_{=\xi(\RW-1)} W^* -(\mu_1 +\beta_B) \right) +\alpha_M =0, 
\end{equation*}
or equivalently, a quadratic expression in $W^*$,
$$\phi_a W^{*2} + \phi_b W^* + \phi_c =0,$$
with $\phi_a$, $\phi_b$ and $\phi_c$ defined as in \eqref{conditions2}. The assumption $\RM>1$ and $\RW >1$ implies $\phi_a >0$ and $\phi_c>0$, hence the parabola opens up and has positive intercept with the y-axis. If $\phi_b >0$, then the vertex of the parabola lies on the left half of the plane and the zeros $W^*_{2,3}$ are not of biological interest. In contrast, if $\phi_b <0$ the vertex of the parabola lies on the right half plane. The condition $\Delta_{\phi}>0$ guarantees the existence of two positive real roots $W^*_{2,3}$.
\end{proof}

\section{How to treat infestations: Four Possible Strategies}
\label{sec:treat_simulax}
Untreated infestations lead to large lice colonies and possibly to secondary bacterial infections~\cite{Cummings2018}. Fig.~\ref{Fig:No_treatment} shows the evolution in time of a lice colony which develops from a small group of adults if untreated for about 6 weeks. For the numerical simulations in Fig.~\ref{Fig:No_treatment} and for all other figures in this section we use, if otherwise not explicitly mentioned, parameter values as indicated in Table~\ref{Table_rates} and the initial conditions $U(0)=0$, $N(0)=0$, $W(0)=4$, $M(0)=4$, $W^B(0)=0$. Such initial values mirror the fact that an initial infestation usually involves less than 10 live lice~\cite{Cummings2018} and it is due to adult lice, which are able to move from host to host (cf. Sec.~\ref{sec:Introduction}).\\
\ \\
In the following we present and compare four different strategies for treatment of head lice infestation, aiming to fast and effective lice eradication. We shall denote a treatment as \textit{effective}, if the infected host is "lice-free" for two weeks (14 days) after the last treatment application. Similar to Laguna and Risau-Gusman \cite{Laguna2011}, we introduce a critical detection threshold and define a host "lice-free" when the egg/live lice population has dropped below this threshold. In all plots in Figs.\ref{Fig:Strategy1}-\ref{Fig:Strategy4} and Fig.~\ref{Fig7}, the solid black curve represents eggs ($U$), whereas the dashed black curve represents live lice (that is, the sum of nymphs and adults). Red and blue dotted lines indicate the detection threshold (assumed here to correspond to 2 eggs/live lice) and applications of treatments, respectively. We assume that lice are discovered about three weeks after initial infestation and that treatments start immediately after detection (day 21).  We start considering lice on an isolated host, who cannot be reinfected while or after being treated. That is, for all $t>0$ we set $\alpha_j (t)\equiv 0 \equiv \beta_j(t)$, $j=W,M$.\\
\ \\
\noindent \textbf{Strategy nr. 1: Classical topical treatments, non-ovocidal.} Shampoos and lotions based on insecticides such as	malathion, pyrethis and its derivatives (e.g. permethrin) kill mature nymphs and adult lice but are mostly non-ovocidal \citep{Speare2007}. It is recommended to apply two-three treatments with shampoos one week apart, the third one being necessary in severe cases~\citep{Speare2007}.
Hence, in our first attempt we shall simulate three applications of a topical non-ovocidal treatment at days 21, 28 and 35. Assuming that the treatment is perfectly working and eliminates no eggs but 100\% of live lice, then the strategy is effective, that is, three treatments are sufficient to get rid of the infestation (Fig.~\ref{Fig:shampoo100}). The timing of application of insecticide-based shampoo relies on the biology of the lice life-cycle. Being non-ovocidal, shampoos do not harm eggs, which will hatch and evolve into new adults, if the gap between the applications is too short. See for example in Fig.~\ref{Fig:shampoo100short} simulations for a shampoo applied three days in a row following detection. Analogously, when treatments are repeated once a month, the adult lice population has time to fully regenerate and the infestation persists after three treatments (Fig.~\ref{Fig:shampoo100toolong}). If the treatment is perfectly working against nymphs and adult lice, then gaps between the applications can be extended to two weeks, and three applications allow to eradicate the infestation (Fig.~\ref{Fig:shampoo100long}).\\
\ \\
Extensive use of topical treatments has led to selection and development of resistant head lice populations \citep{Yoon2003,Feldmeier2012} so that no shampoo nor lotion truly kills 100\% of live lice.  In an experimental study, Yoon et al.~\citep{Yoon2003} showed that 13-87\% lice were resistant to permethrin. Simulations in Fig.~\ref{Fig:shampoo60res} indicate that in case of 40\% resistant lice, the recommended "three times in two weeks" strategy is not sufficient to eradicate the infestation. Feldmeier et al.~\citep{Feldmeier2012} suggested to treat resistant lice with dimeticones (see Strategy 4).\\
\ \\

\begin{figure}[t]
	\centering
		\begin{subfigure}{0.47\textwidth}
		\includegraphics[width=\textwidth]{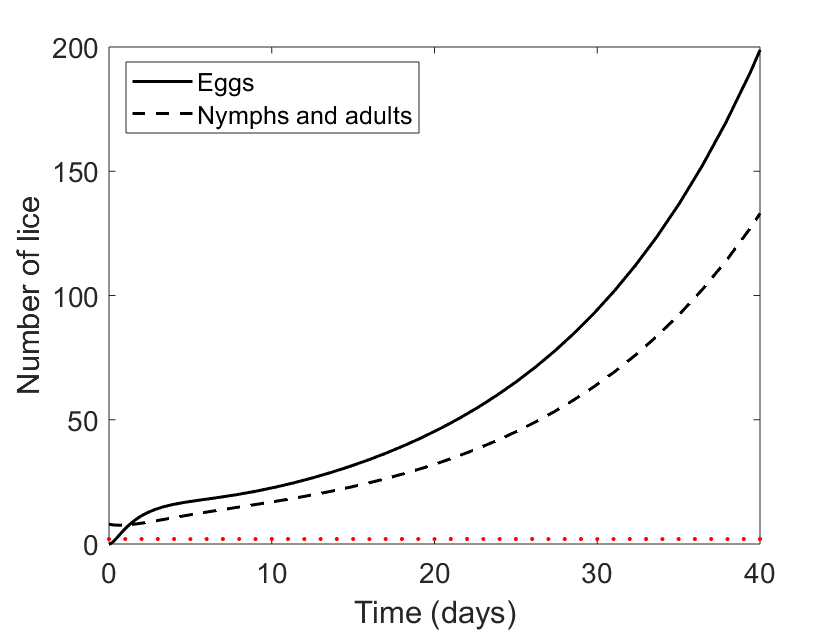}
		\caption{}
		\label{Fig:No_treatment}
	\end{subfigure}
	\begin{subfigure}{0.47\textwidth}
		\includegraphics[width=\textwidth]{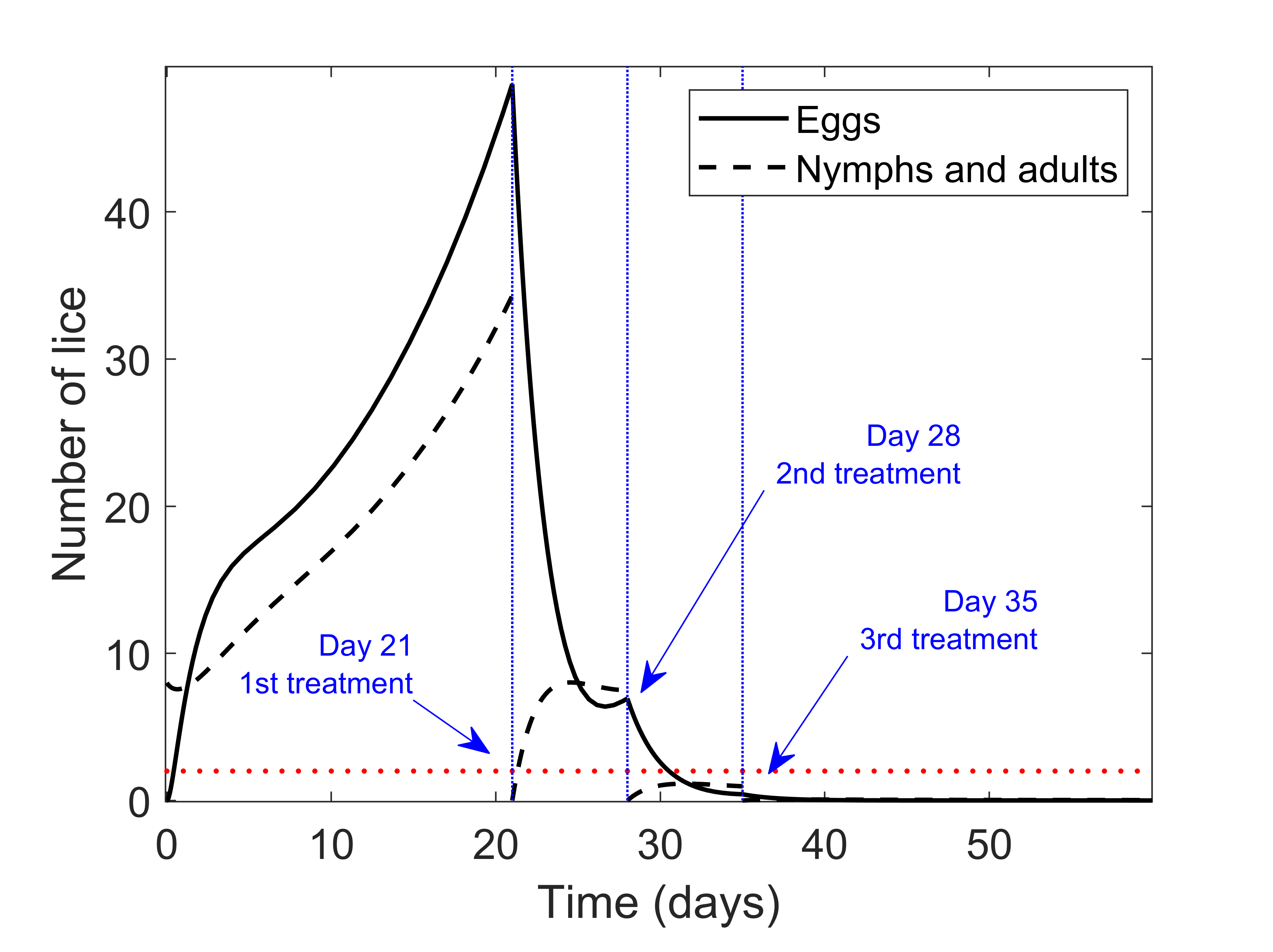}
		\caption{}
		\label{Fig:shampoo100}
	\end{subfigure}
	\begin{subfigure}{0.47\textwidth}
		\includegraphics[width=\textwidth]{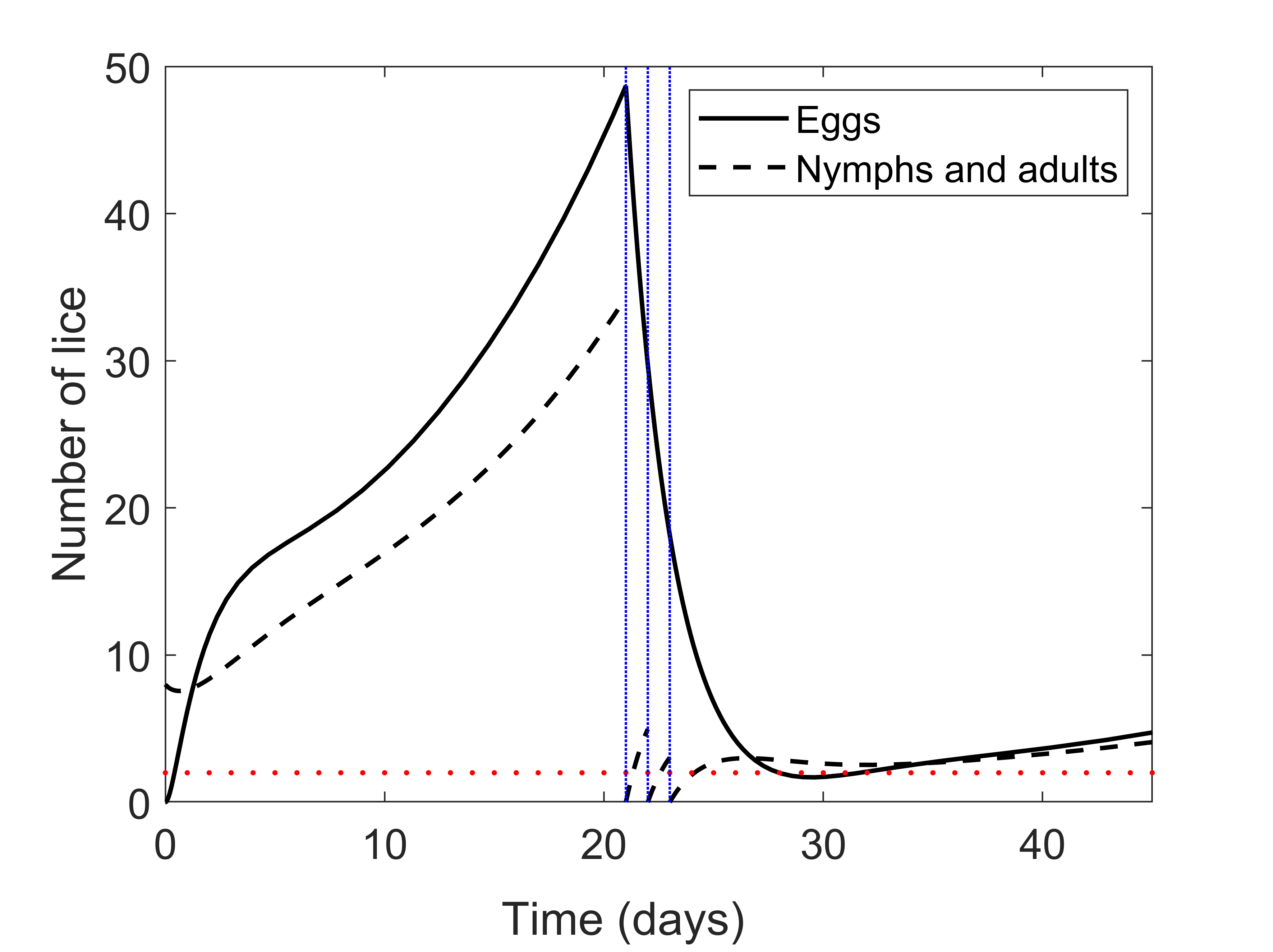}
		\caption{}
		\label{Fig:shampoo100short}
	\end{subfigure}	
	\begin{subfigure}{0.47\textwidth}
	\includegraphics[width=\textwidth]{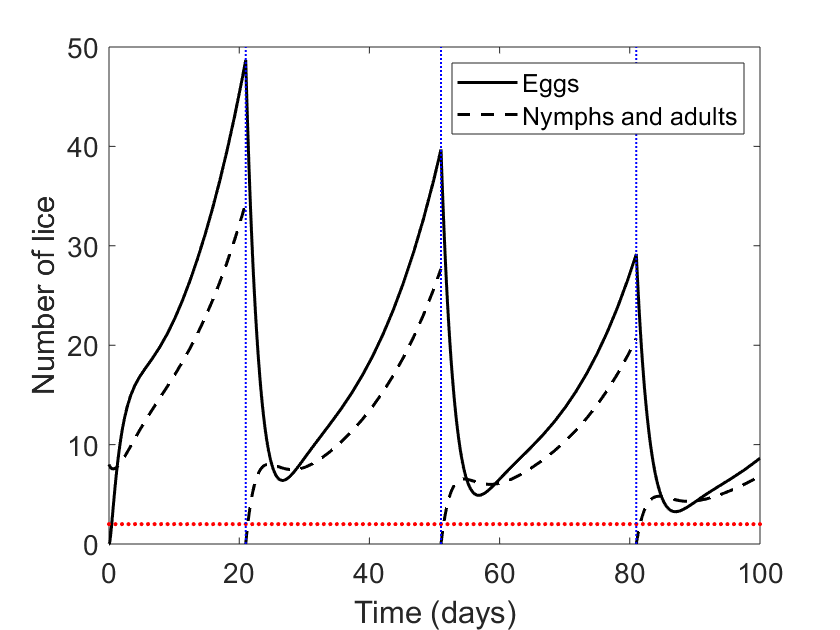}
	\caption{}
	\label{Fig:shampoo100toolong}			
\end{subfigure}
	\begin{subfigure}{0.47\textwidth}
		\includegraphics[width=\textwidth]{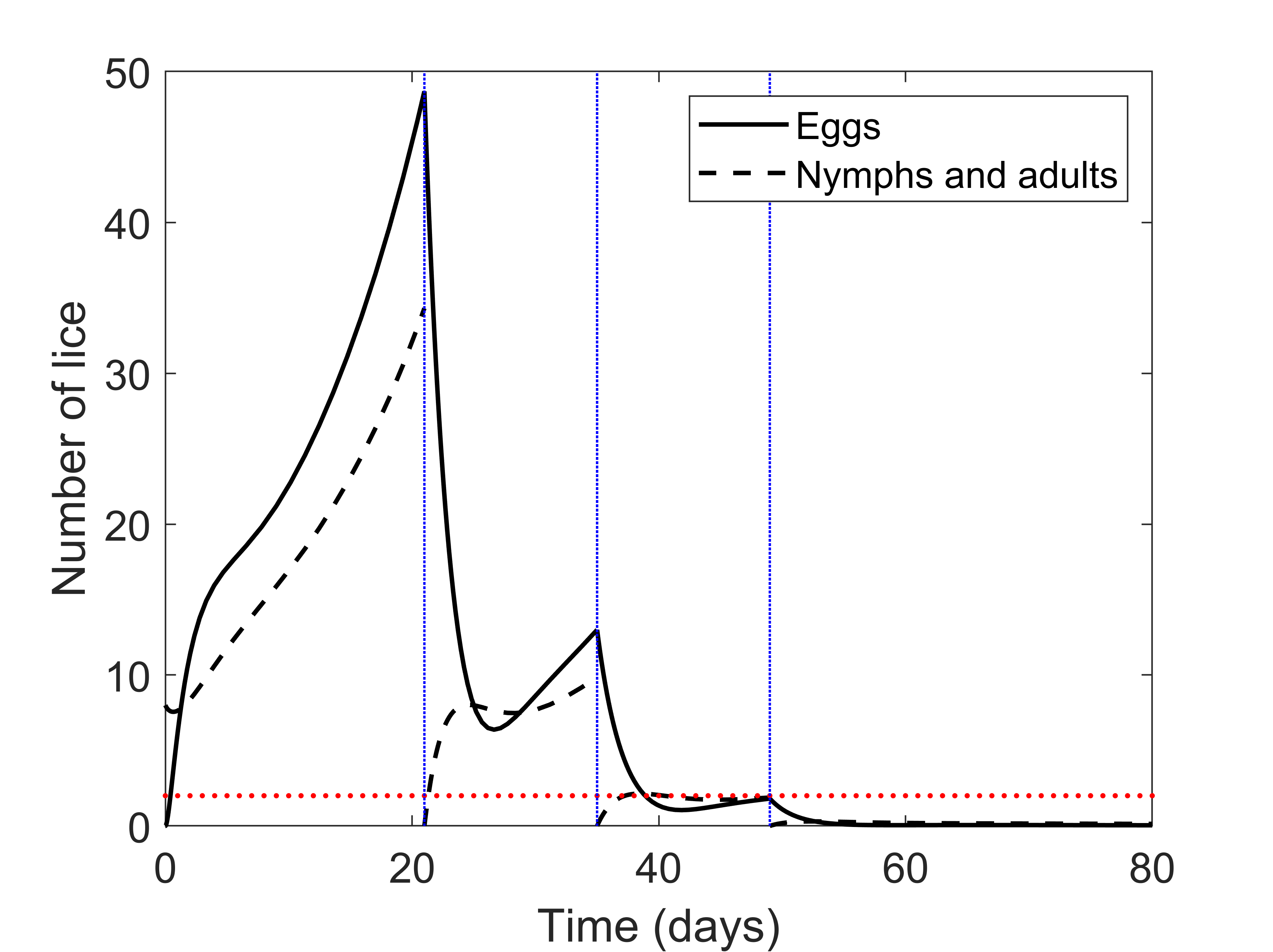}
		\caption{}
		\label{Fig:shampoo100long}			
	\end{subfigure}
	\begin{subfigure}{0.47\textwidth}
		\includegraphics[width=\textwidth]{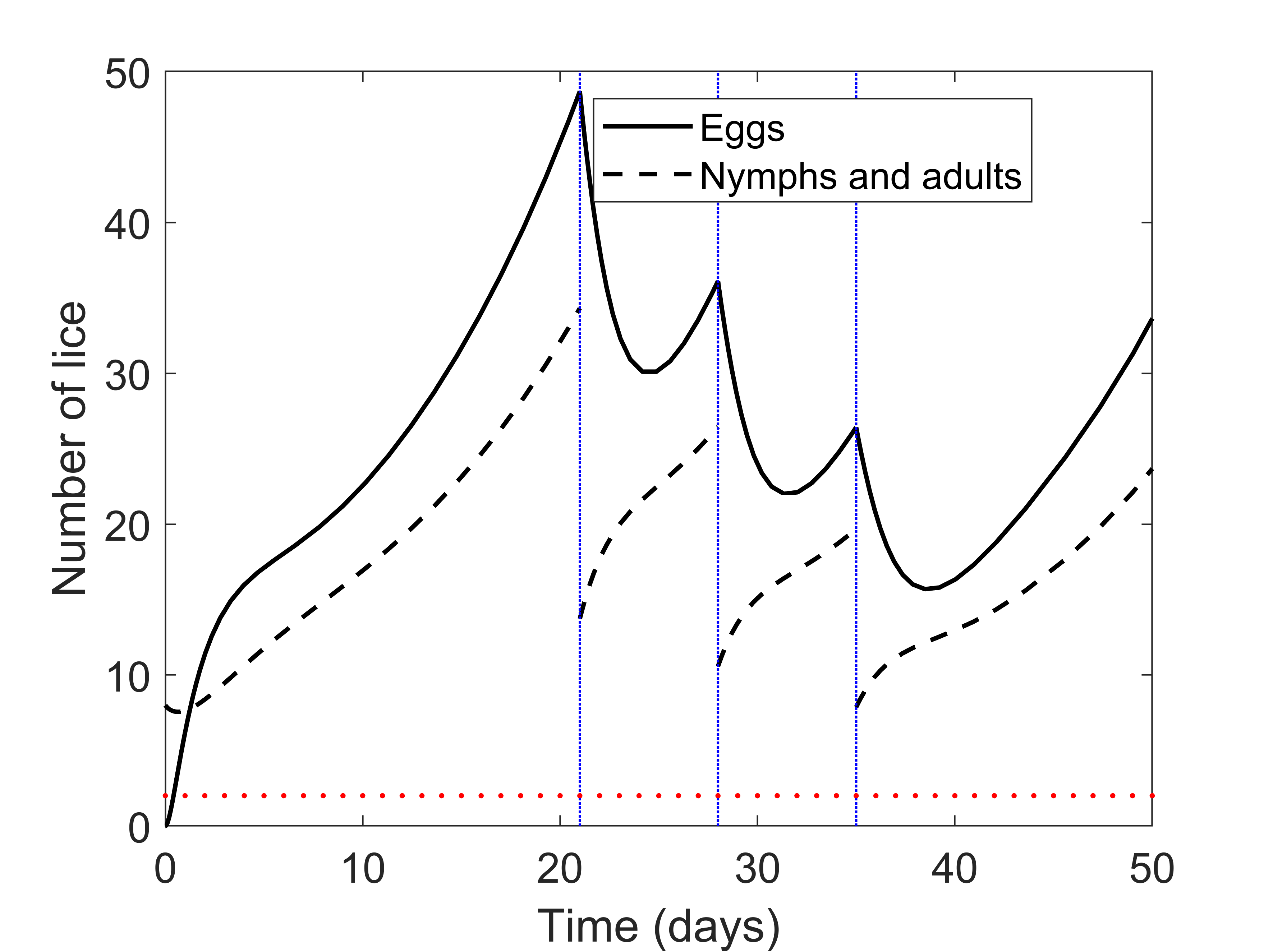}
		\caption{}
		\label{Fig:shampoo60res}
	\end{subfigure}
\caption{Evolution in time of a lice colony which develops from a small group of adults. (a) Lice colony untreated for 40 days. (b-f) Strategy nr.1. Starting from day 21 (first treatment) since the beginning of the infection, two further applications with a fully working topical treatment (killing 100\% of nymphs/adults) are repeated (b) after 7 days (day 28 and 35), (c) after 1 day  (day 22 and 23), (d) after 30 days (day 51 and 81), and (e) after 14 days from each previous treatment (day 35 and 49). (f) Topical treatments applied as in (b) assuming 40\% resistance in nymphs and adult lice.}
	\label{Fig:Strategy1}
\end{figure}

\noindent \textbf{Strategy nr. 2: Conditioner and Combing method.} Less expensive than topical treatments and non-chemical, wet combing is an optimal method for detection of head lice infestations~\cite{Feldmeier2012}. Health care institutions recommend that the hair is divided into small sections and each section is combed completely, repeating the combing procedure every one-two days until no lice are detected for 10 consecutive days \citep{WADepHealth}. It is difficult to assess and quantify the efficacy of wet combing from previous scientific studies~\cite{Feldmeier2014} as this depends on a number of factors, including the nature of the comb~\cite{Gallardo2013,Speare2002}. Therefore, for numerical simulations we consider here three scenarios: (i) low effectiveness (combing removes 20\% of live lice/eggs), (ii) moderate effectiveness (combing removes 50\% of live lice/eggs), and (iii) high effectiveness (combing removes 80\% of live lice/eggs). As recommended in~\citep{WADepHealth} we apply combing every second day until no lice/eggs are detected (meaning that both populations dropped below the detection threshold) and observe how the lice population evolves in the following two weeks. Simulations in Fig.~\ref{Fig:Strategy2} show that the duration of the treatment importantly depends on the effectiveness of the combing procedure, varying from 25 applications (Fig.~\ref{Fig:Strategy2}a) when combs remove only 20\% lice/eggs, to 2 applications (Fig.~\ref{Fig:Strategy2}e) when combs remove 80\% lice/eggs. The duration of the treatment can be reduced, in particular when the treatment eliminates only 20\% of live lice/eggs, by combing the hair every day instead of every second day (Fig.~\ref{Fig:Strategy2}(b,d,f)). Notice that interrupting the treatment here is not necessarily implying that the treatment strategy was effective. Indeed, in all cases considered in Fig.~\ref{Fig:Strategy2} the host is lice free for a few days, but the lice population grows above the detection threshold within 14 days from the last treatment, with borderline values in some cases (Fig.~\ref{Fig:Strategy2}d).\\
\ \\
	
\begin{figure}[h!]
	\centering
		\begin{subfigure}{0.47\textwidth}
		\includegraphics[width=\textwidth]{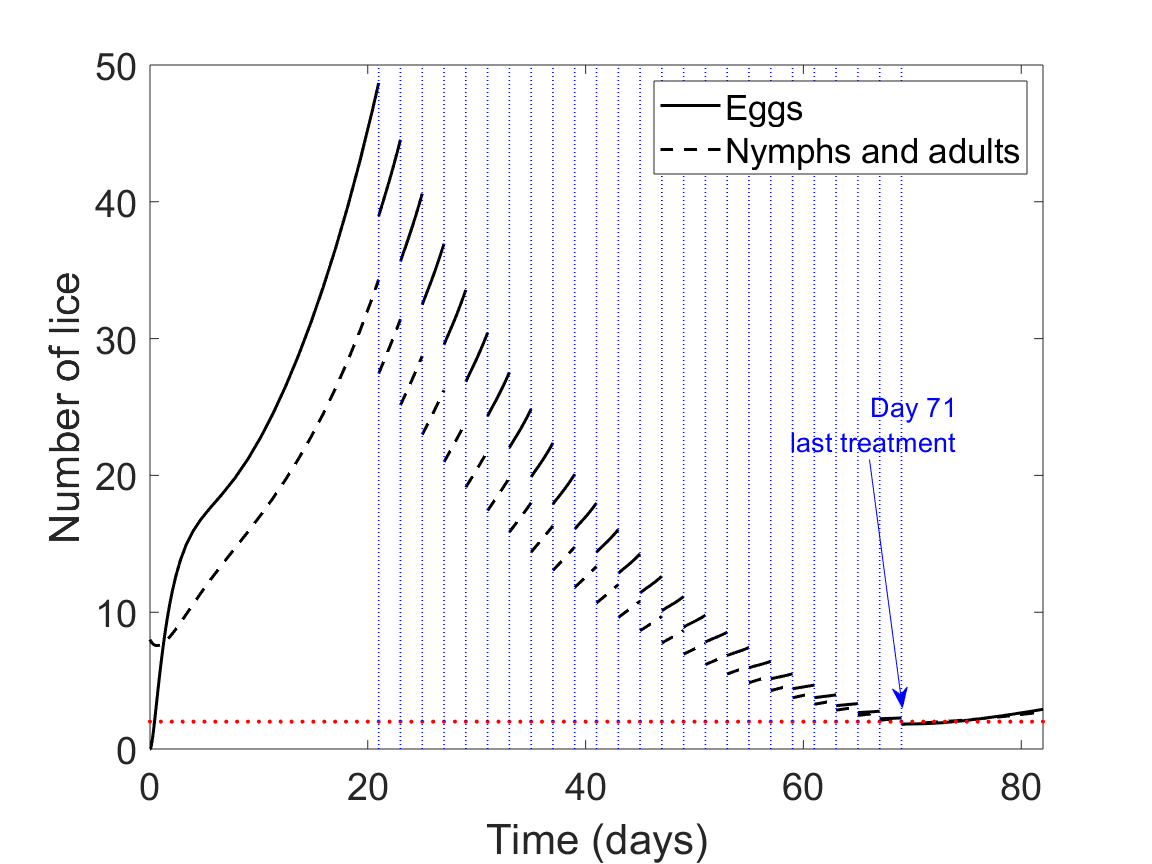}
		\caption{}
		\label{Fig:str2a}			
	\end{subfigure}
	\begin{subfigure}{0.47\textwidth}
		\includegraphics[width=\textwidth]{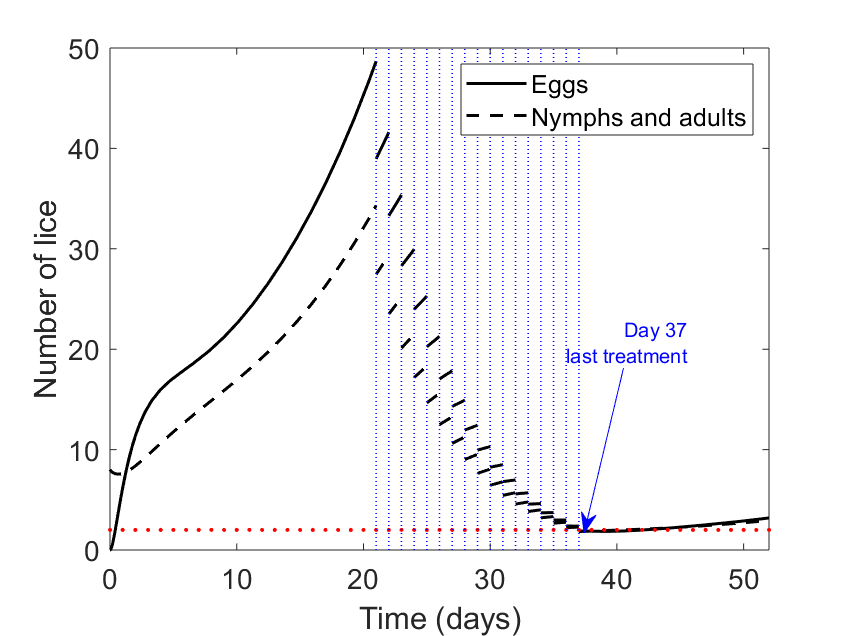}
		\caption{}
		\label{Fig:str2b}
	\end{subfigure}
	\begin{subfigure}{0.47\textwidth}
	\includegraphics[width=\textwidth]{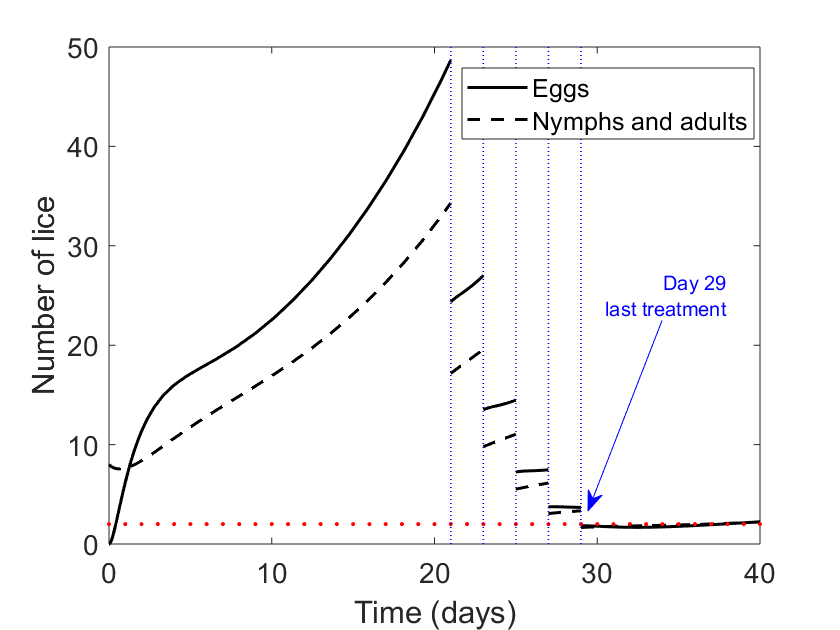}
	\caption{}
	\label{Fig:str2c}			
\end{subfigure}
\begin{subfigure}{0.47\textwidth}
	\includegraphics[width=\textwidth]{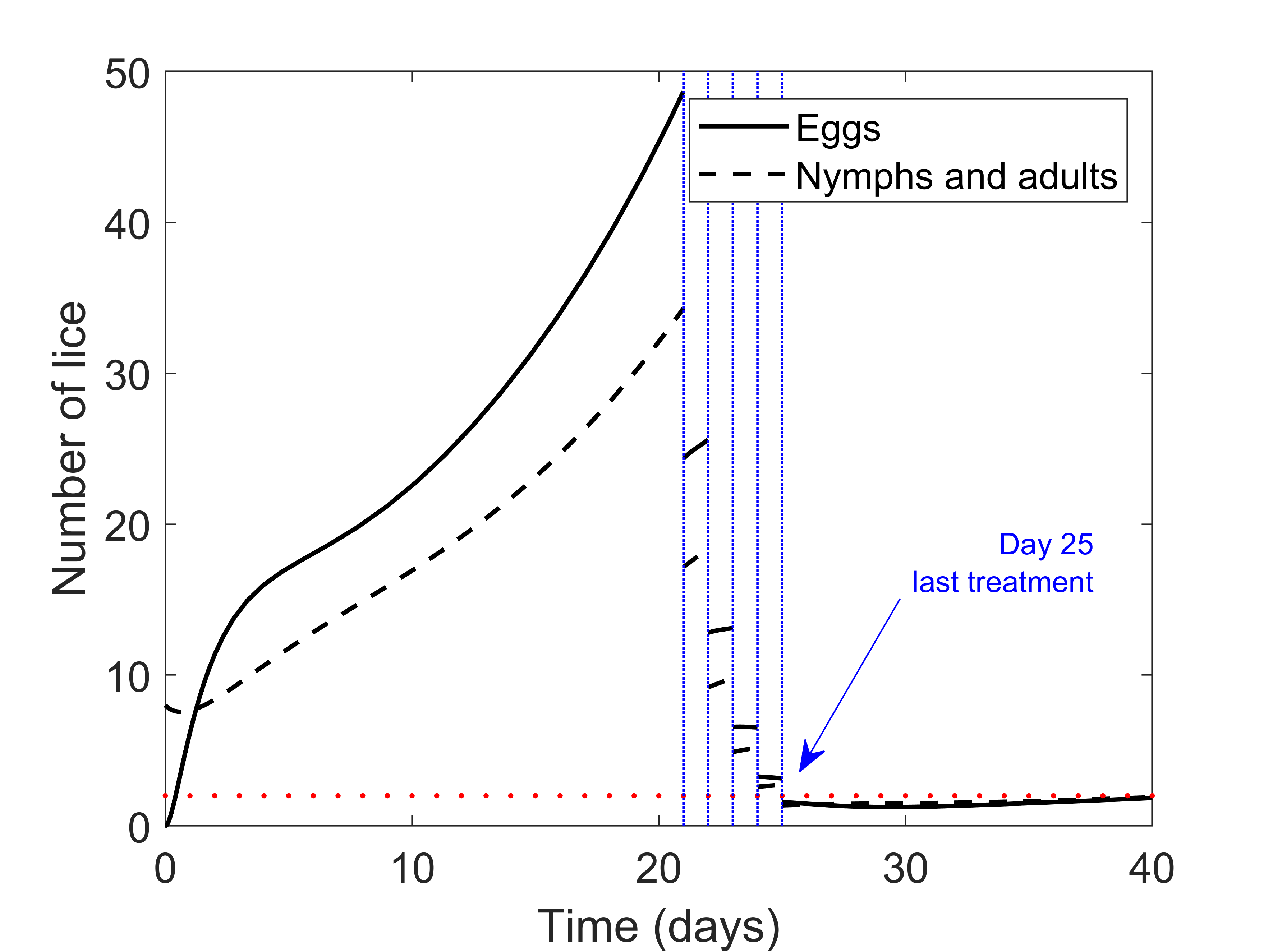}
	\caption{}
	\label{Fig:str2d}
\end{subfigure}
	\begin{subfigure}{0.47\textwidth}
	\includegraphics[width=\textwidth]{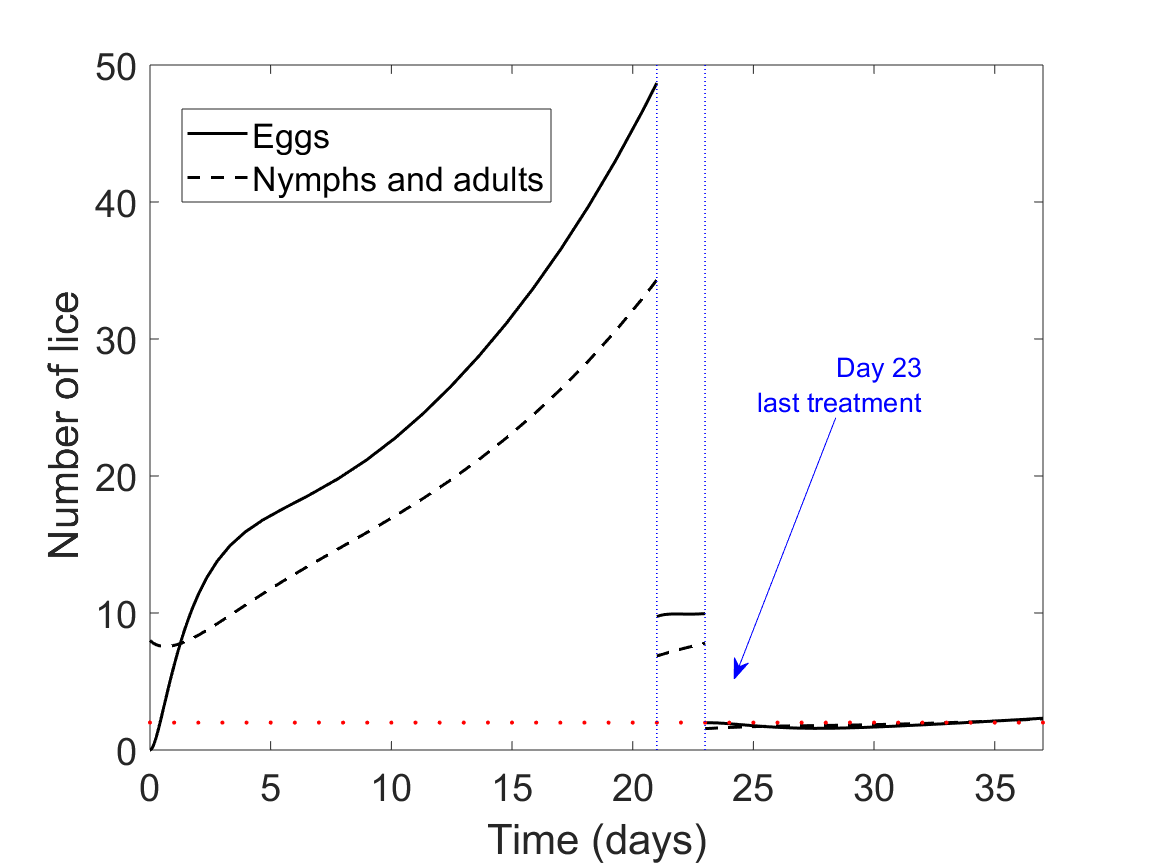}
	\caption{}
	\label{Fig:str2e}			
\end{subfigure}
\begin{subfigure}{0.47\textwidth}
	\includegraphics[width=\textwidth]{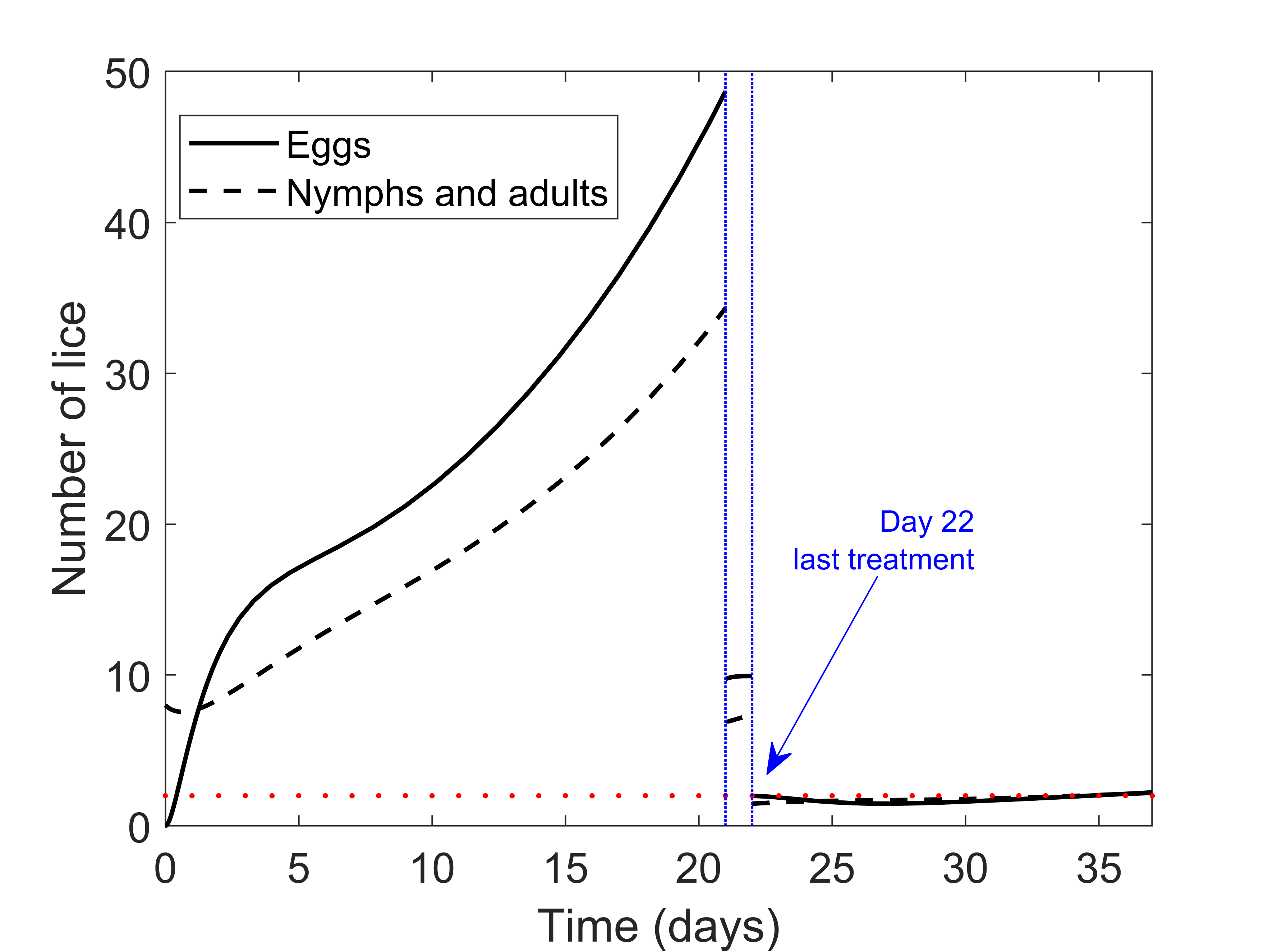}
	\caption{}
	\label{Fig:str2f}
\end{subfigure}
	\caption{Strategy nr. 2. Evolution in time of a lice colony which develops from a small group of adults and is treated with conditioners and combs. Starting from day 21 (first treatment) since the beginning of the infection, lice are treated every second day (a,c,e) or every day (b,d,f). Web combing it is assumed to eliminate (a,b) 20\%, (c,d) 50\% or (e,f) 80\% of lice/eggs.}
	\label{Fig:Strategy2}
\end{figure}	

\noindent \textbf{Strategy nr. 3: Combined treatments (shampoos and combs).} 
In strategies nr.1 and nr.2, shampoos and wet combing were applied separately. Guidelines for head lice treatment have previously suggested to combine different products, applying a shampoo every seven days and a lice comb every two days between one shampoo and the next~\cite{queensland2019}. We simulate (Fig.~\ref{Fig4a}) the effect of such a combined strategy, assuming that starting from detection lice are treated three times with shampoo (day 21, 28 and 35) and further with low/moderately effective wet combining (killing 20/50\% of live lice/eggs) every second day between two shampoos (days 23, 25, 30 and 32).
This alternate treatments strategy is effective when combing is removing 50\% of live lice and nits, whereas it is not when combing effectiveness is low, compare Figs.~\ref{Fig:Strategy3}(a,b). If combing effectiveness is low, but the hair is combed more often, e.g. four times, between two topical treatments, the strategy could also be considered effective as after the last shampoo the lice population remains for two weeks below the detection threshold (Fig.~\ref{Fig4c}).\\
\ \\ 

\begin{figure}[h!]
	\centering
	\begin{subfigure}{0.32\textwidth}
		\includegraphics[width=\textwidth]{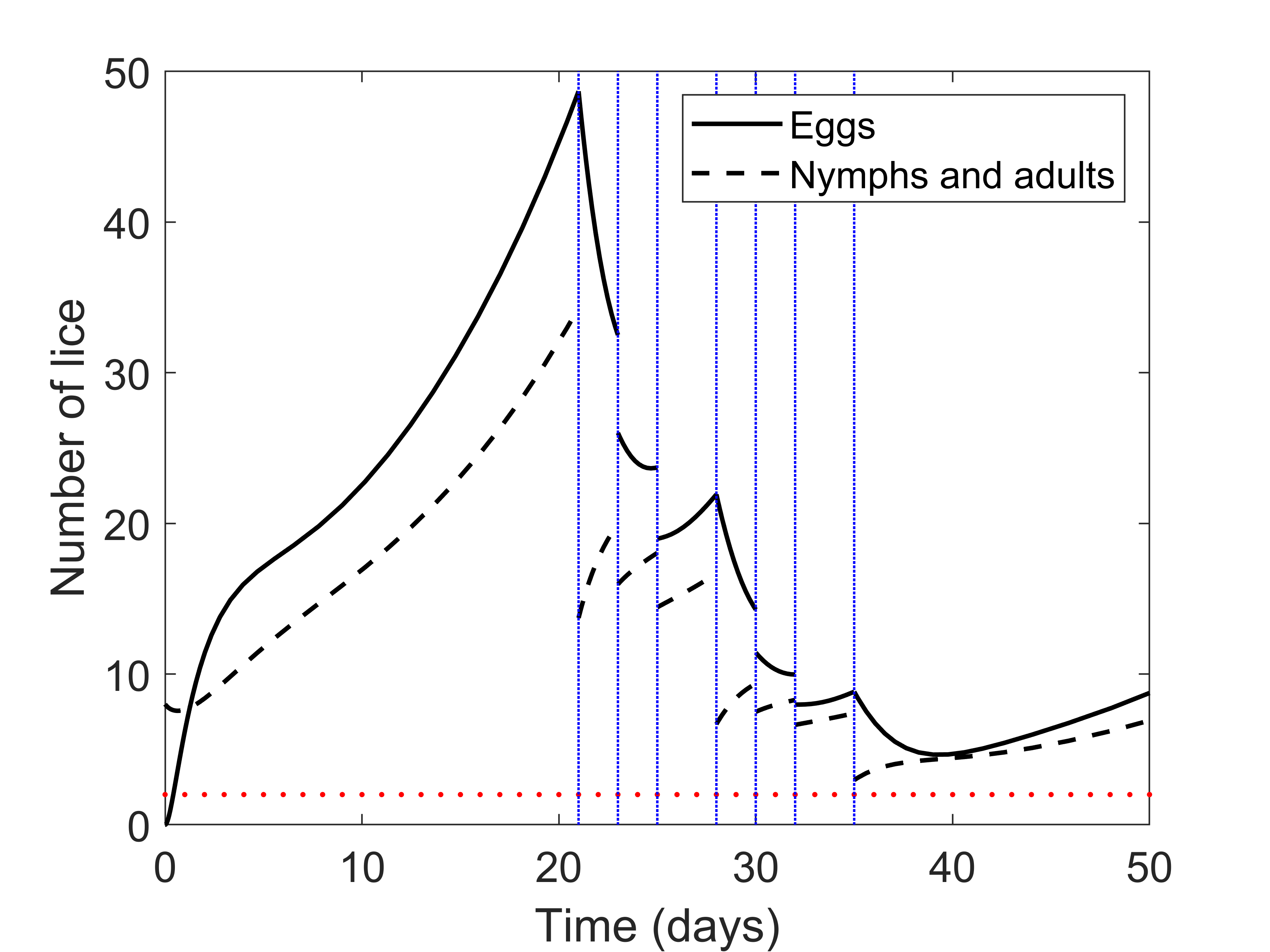}
		\caption{}
		\label{Fig4a}			
	\end{subfigure}
	\begin{subfigure}{0.32\textwidth}
	\includegraphics[width=\textwidth]{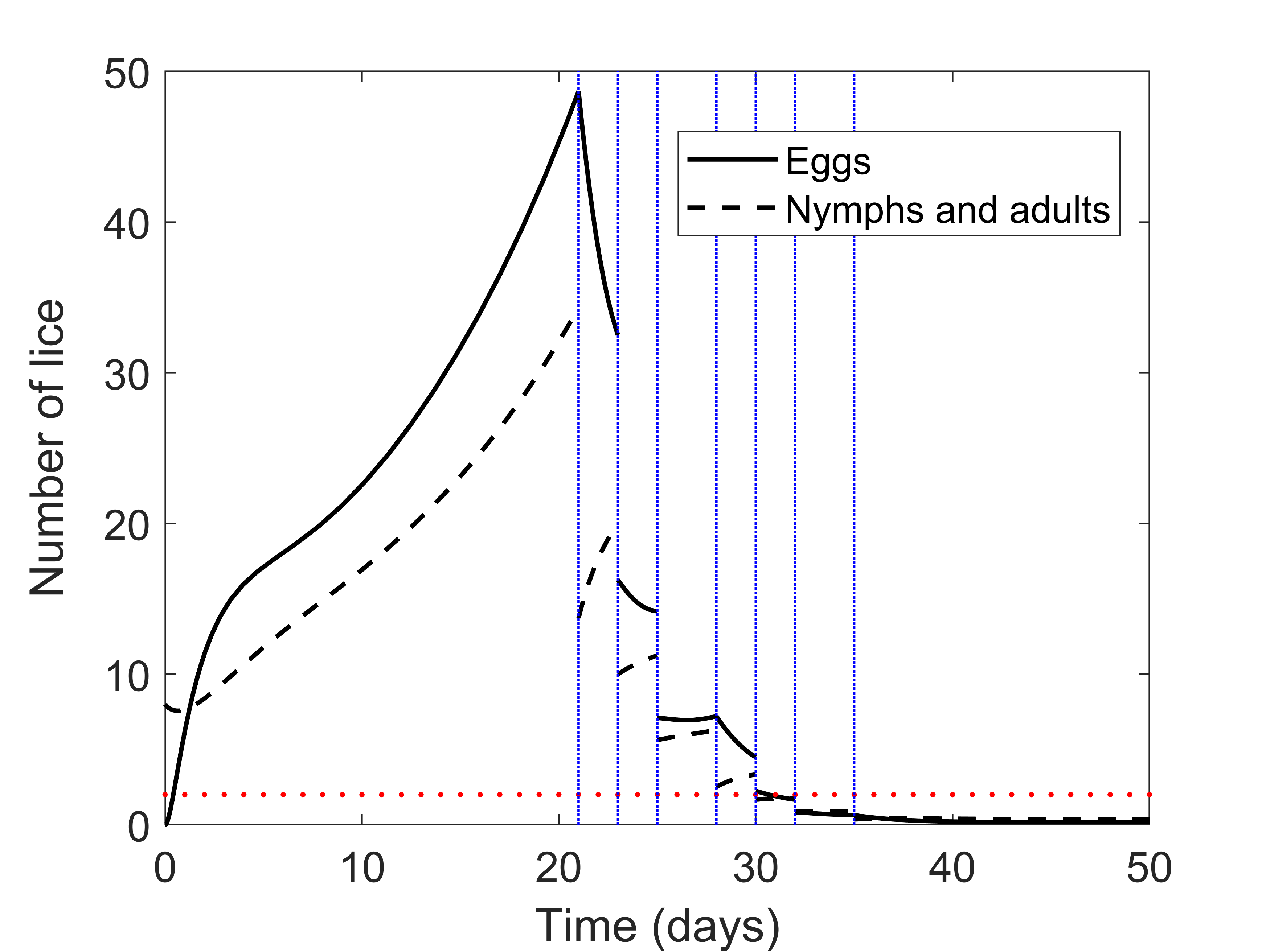}
	\caption{}
	\label{Fig4b}			
\end{subfigure}
\begin{subfigure}{0.32\textwidth}
	\includegraphics[width=\textwidth]{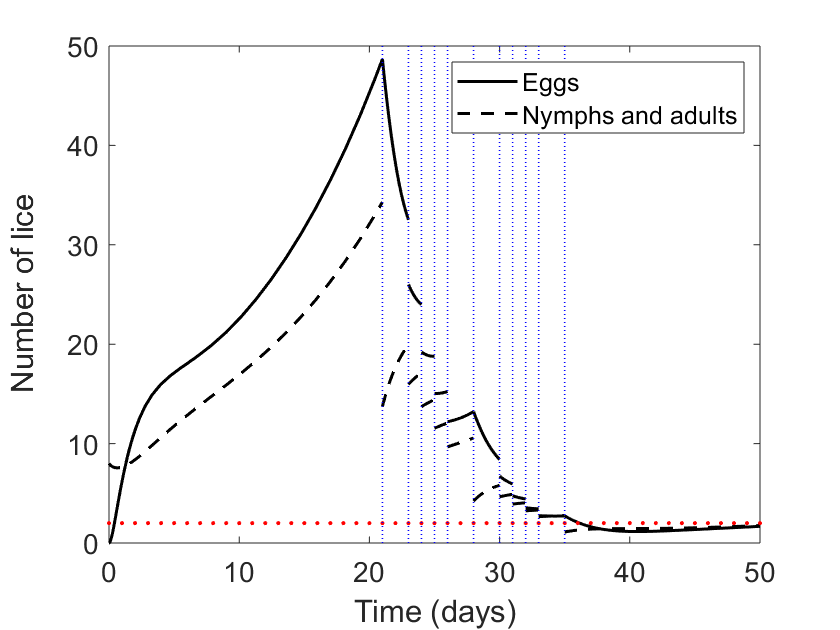}
	\caption{}
	\label{Fig4c}
\end{subfigure}
	\caption{Strategy nr. 3. Evolution in time of a lice colony which develops from a small group of adults and is treated with non-ovocidal topical treatments alternated to wet combing. Starting from detection lice are treated with shampoo at days 21, 28 and 35 (killing 60\% of live lice) in alternation with wet combining (a,b) on days 23, 25, 30 and 32, respectively (c) on days 23, 24, 25, 26, 30, 31, 32, 33. Combing effectiveness was assumed to be (a,c) low  (eliminating 20\% of live lice/eggs), or (b) moderate (eliminating 50\% of live lice/eggs).}
	\label{Fig:Strategy3}
\end{figure}

\noindent \textbf{Strategy nr. 4: Dimeticones-based treatments.} Dimeticones are silicone oils which have been recently employed in anti-head lice compounds. When applied on a louse, dimeticones enter into the spiracles, interrupt oxygen supply and lead to rapid death of the insect \citep{Heukelbach2008}. Two kind of dimeticones have been recently studied (see \cite{Feldmeier2012} and references thereof): (i) Hedrin$^{\textregistered}$, 4\% dimeticones solution, which showed 70\%-92\% efficacy on treating lice infestations, despite being ineffective on eggs, and (ii) NYDA$^{\textregistered}$, a combination of two dimeticones which is also ovocidal (95-100\% eggs killed). 
Being non-ovocidal, treatments with Hedrin$^{\textregistered}$  can be associated to the previously simulated strategies nr. 1 and nr. 3 (the latter, if combined with wet combing). In contrast, NYDA$^{\textregistered}$  was proposed as a good candidate for a two-application treatment, with a recommended second treatment 8-10 days after the first one~\cite{Cummings2018, Feldmeier2012}. Assuming that NYDA$^{\textregistered}$  eliminates 80\% of nymphs and adults and 97\% of eggs, we simulate two treatments with NYDA$^{\textregistered}$, the first at detection (day 21) and the second at day 31. However, with this treatment schedule the lice population will grow beyond the detection threshold in less than one week (Fig.~\ref{Fig6A}), suggesting that the treatment did not work. The strategy becomes effective when we anticipate the second treatment to day 25 (Fig.~\ref{Fig6B}) because, nymphs and adults being reduced by NYDA$^{\textregistered}$ to 20\% of their values, the population growth slows down importantly and lice remain under the detection threshold for two weeks.

\begin{figure}[h!]
	\centering
	\begin{subfigure}{0.48\textwidth}
		\includegraphics[width=\textwidth]{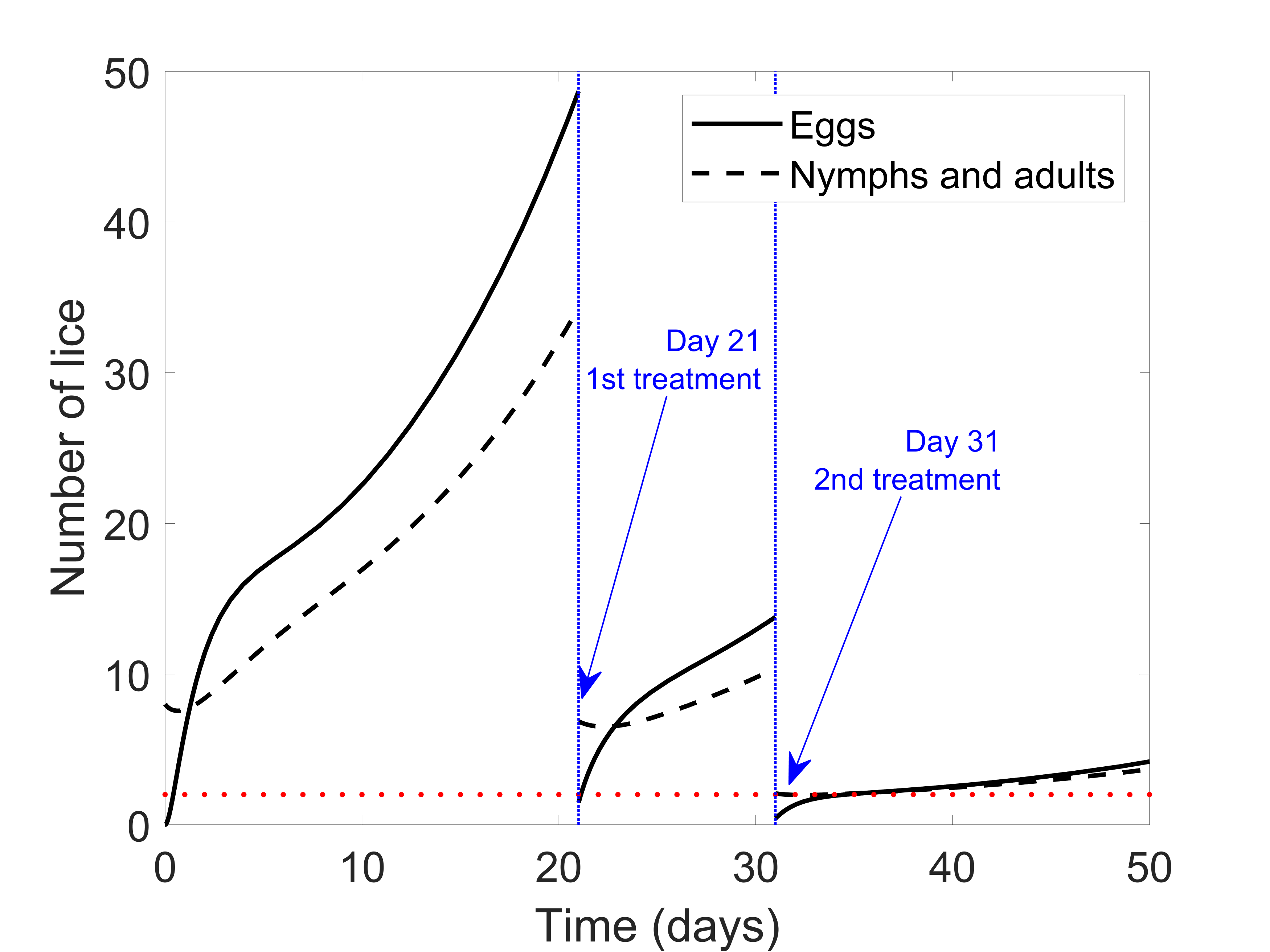}
		\caption{}
		\label{Fig6A}			
	\end{subfigure}
	\begin{subfigure}{0.48\textwidth}
		\includegraphics[width=\textwidth]{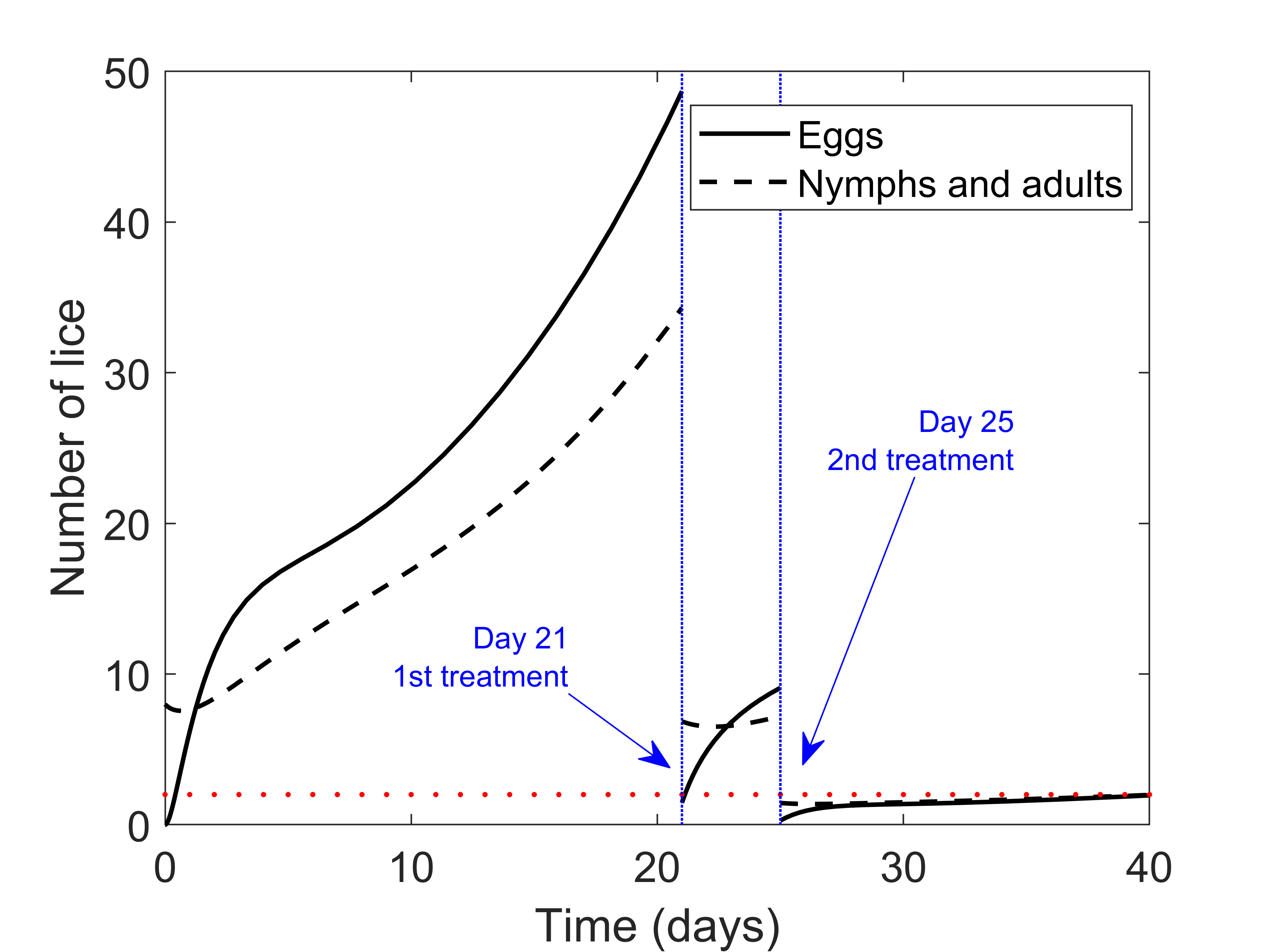}
		\caption{}
		\label{Fig6B}			
	\end{subfigure}
	\caption{Strategy nr. 4. Evolution in time of a lice colony which develops from a small group of adults and is treated with dimeticones (NYDA$^{\textregistered}$). At each application, NYDA$^{\textregistered}$ is assumed to eliminate 80\% of nymphs and adults and 97\% of present eggs. Starting from day 21 (first treatment) since the beginning of the infection, NYDA$^{\textregistered}$ application is repeated (a) after 10 days (day 31) and (b) after 4 days (day 25).}
	\label{Fig:Strategy4}
\end{figure}

\ \\
\noindent One might wonder to what extent our results depend on the timing of the first treatment (assumed to be day 21 in Figs.~\ref{Fig:Strategy1}--\ref{Fig:Strategy4}), or in other words on the population size of eggs ($U_d$) or live lice ($L_d$) at detection.
We let now $U_d$ (respectively, $L_d$) free to vary in the interval [0,200] (respectively, [0,100]), and consider the above presented treatment strategies with variable population size at the time of the first treatment.  We shall distinguish regions of the plane ($U_d,L_d$) indicating weather zero (light yellow), one (dark yellow), two (orange), three (red), or at least four (dark red) treatments are needed to consider the strategy effective. If the initial lice population is very small, because of the local attractiveness property of the lice-free equilibrium (Theorem~\ref{prop:LFE}) the lice population dies out without intervention, hence no treatment needs to be applied. When a large amount of eggs and/or live lice is present then at least one treatment is necessary to eradicate the infestation. The number of treatments necessary to define the strategy effective depends on the applied product and on the scheduling. We visualize in Fig.~\ref{Fig:DependInitCond} lice treatments with
(a) highly effective topical products eliminating 90\% of life lice applied once every 7 days; (b) moderately effective topical products eliminating 60\% of life lice, applied once every 7 days; (c) moderately effective comb eliminating 50\% eggs/lice, applied every second day; (d) NYDA$^{\textregistered}$ applied every 9 days and (e) NYDA$^{\textregistered}$ applied every 4 days. Whereas at most three applications of topical products would be effective in most cases when the product is killing 90\% of nymphs and adult lice (Fig.~\ref{Fig:DependInitCond}a), four or more treatments become necessary to eradicate the infestation when moderately effective shampoos or combs are used (Fig.~\ref{Fig:DependInitCond}(b,c)).
As NYDA$^{\textregistered}$ is assumed to be very effective against eggs, the number of applications necessary to eradicate the infestation depends essentially on the number $L_d$ of nymphs and adult lice at detection. Three applications 9 days apart from each other allow to eliminate lice as long as $L_d<60$ (Fig.~\ref{Fig6d_NT}). In considering strategy nr. 4 we found that NYDA$^{\textregistered}$ becomes more efficient if applied every 4 days, rather than every 8-10 (Fig.~\ref{Fig:Strategy4}b). Fig.~\ref{Fig6e_NT} confirms our findings indicating that at most three NYDA$^{\textregistered}$ applications 4-days apart are sufficient to eradicate the infestation for 
$L_d\in [0,100]$.

\begin{figure}[h!]
	\centering
	\begin{subfigure}{0.32\textwidth}
		\includegraphics[width=\textwidth]{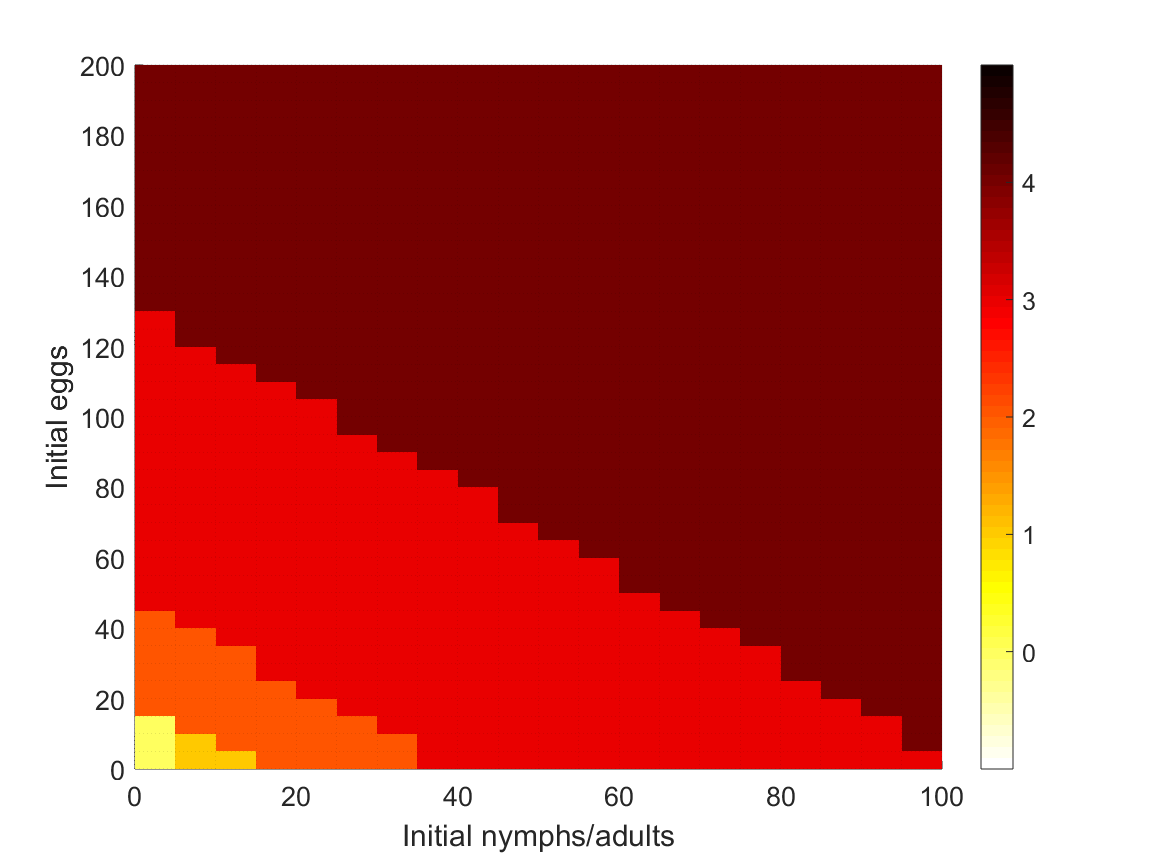}
		\caption{}
		\label{Fig6a_NT}			
	\end{subfigure}
	\begin{subfigure}{0.32\textwidth}
		\includegraphics[width=\textwidth]{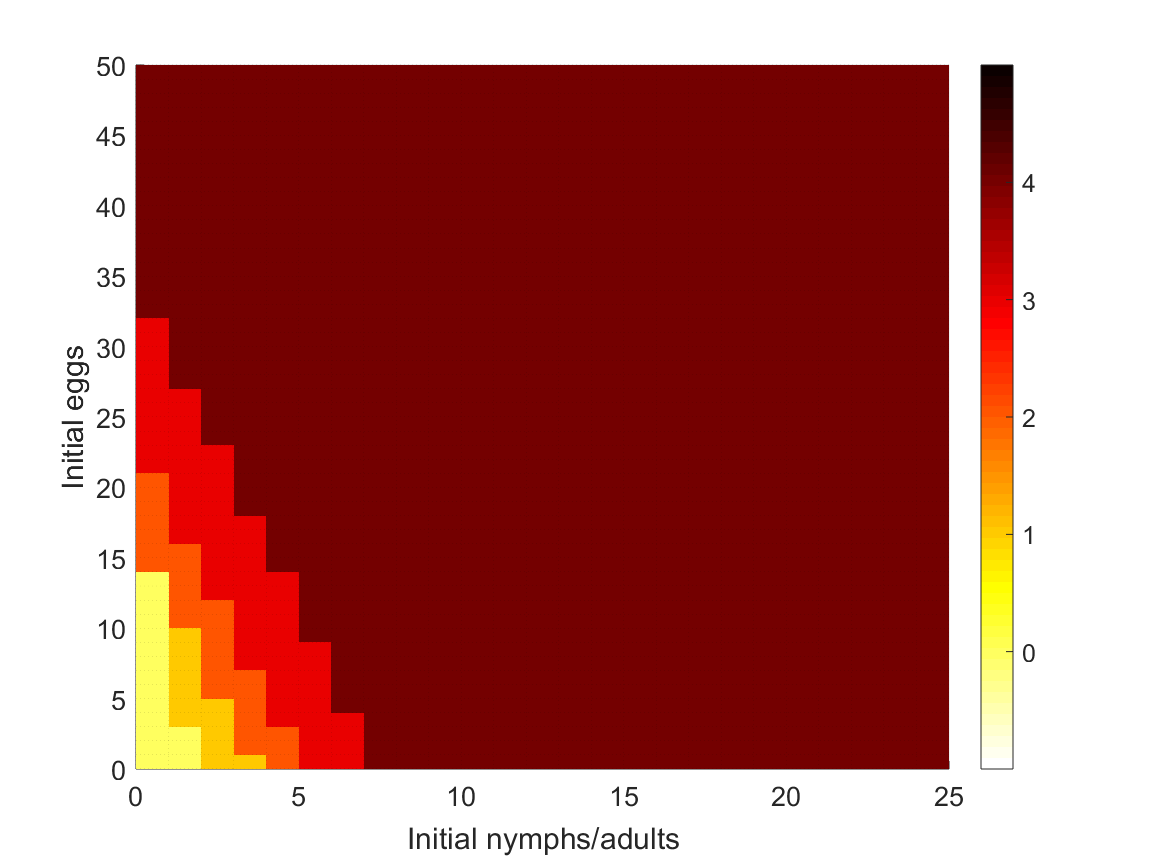}
		\caption{}
		\label{Fig6b_NT}			
	\end{subfigure}
\begin{subfigure}{0.32\textwidth}
	\includegraphics[width=\textwidth]{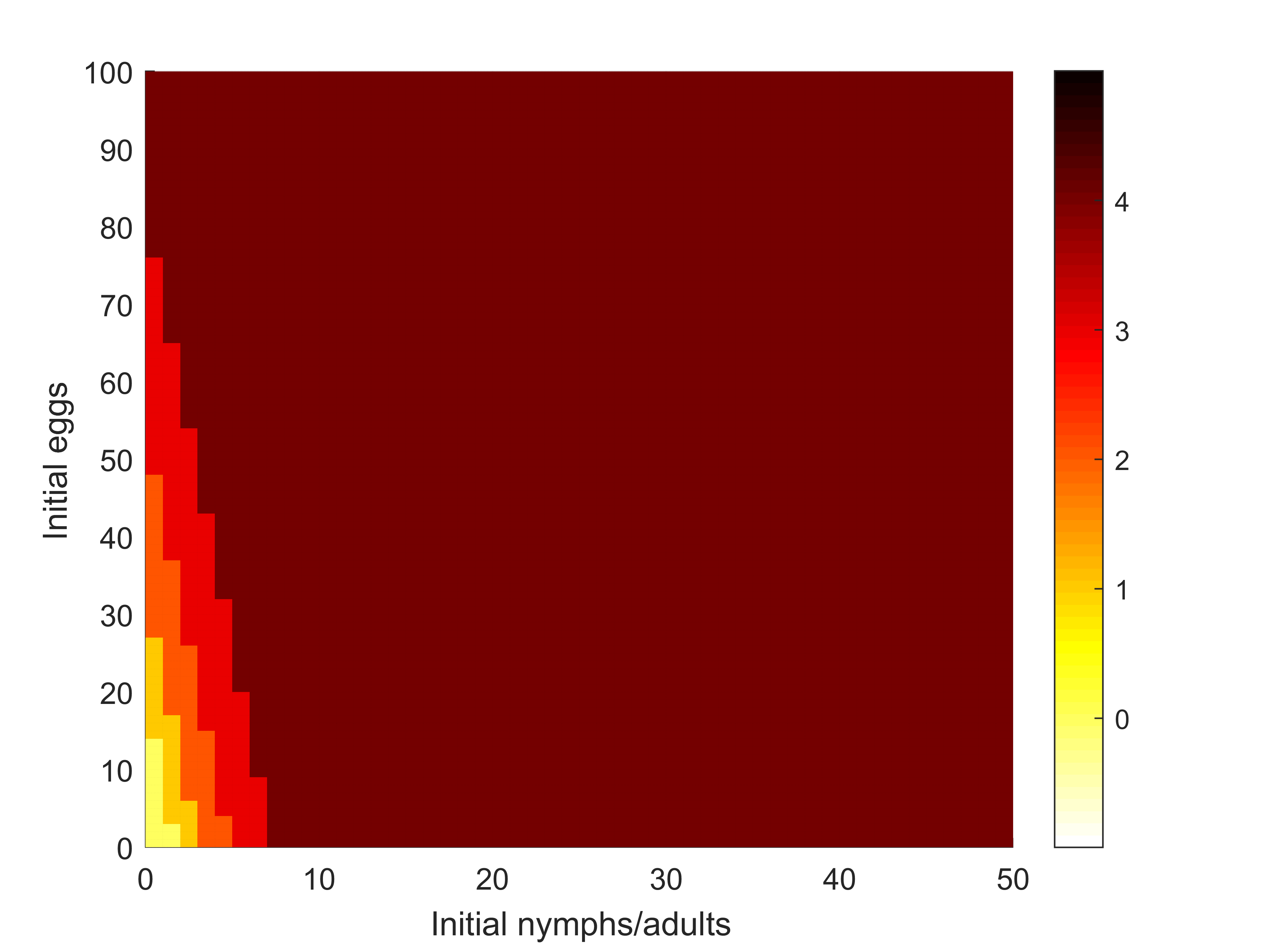}
	\caption{}
	\label{Fig6c_NT}			
\end{subfigure}
\begin{subfigure}{0.48\textwidth}
	\includegraphics[width=\textwidth]{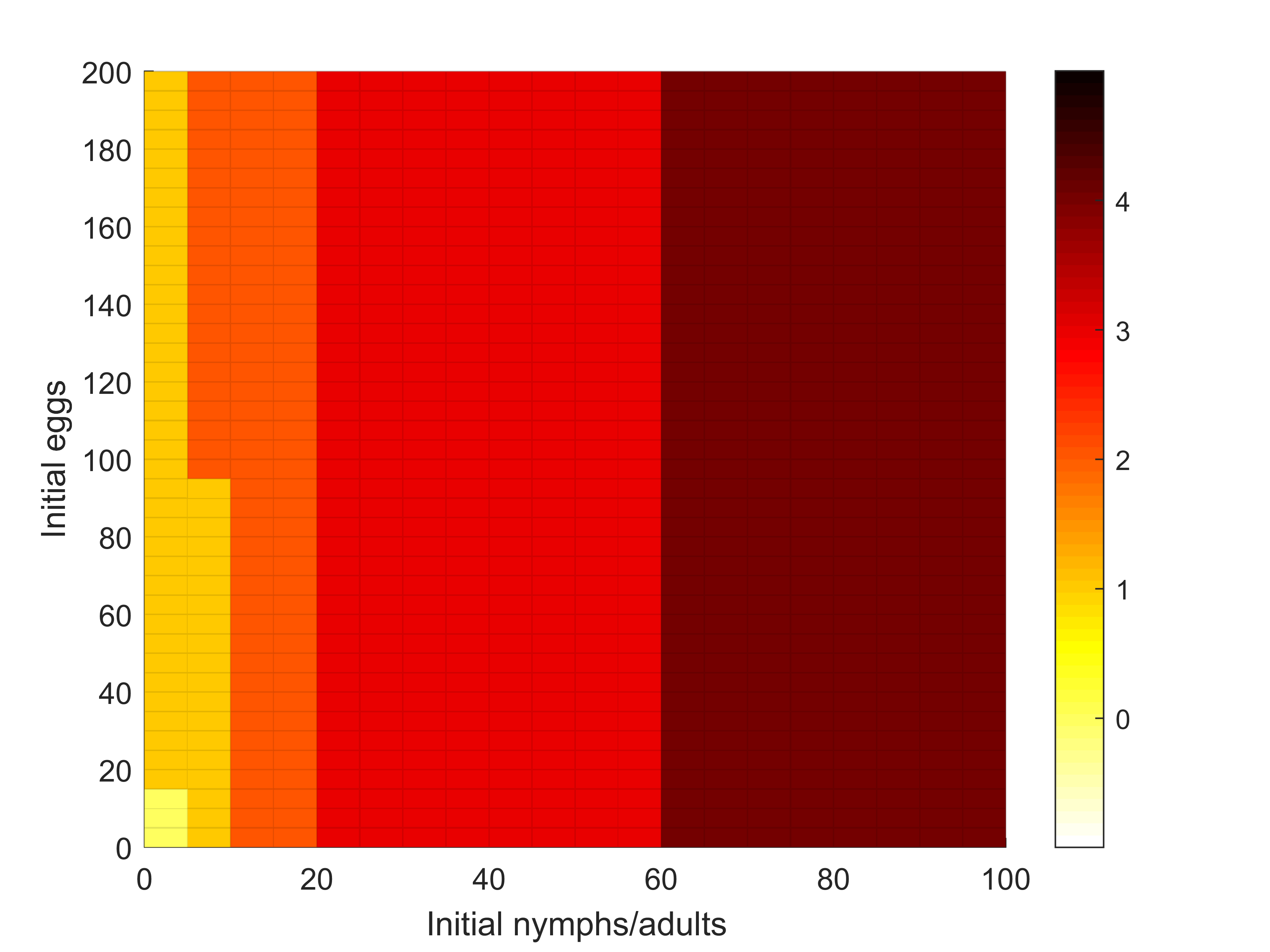}
	\caption{}
	\label{Fig6d_NT}			
\end{subfigure}
\begin{subfigure}{0.48\textwidth}
	\includegraphics[width=\textwidth]{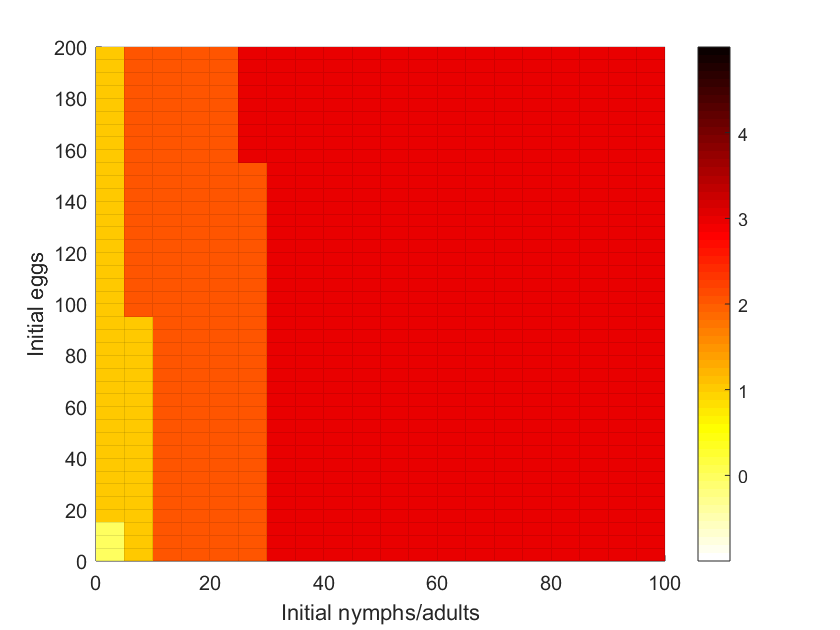}
	\caption{}
	\label{Fig6e_NT}			
\end{subfigure}
	\caption{The severity of a lice infestation at detection affects the number of treatment applications necessary for lice elimination. The panels visualize therapy with: (a) topical treatment eliminating 90\% nymphs/adults, applied once every 7 days; (b) topical treatment eliminating 60\% nymphs/adults, applied once every 7 days; (c) combing eliminating 50\% eggs/lice, applied every second day; (d) NYDA$^{\textregistered}$ applied every 9 days; (e) NYDA$^{\textregistered}$ applied every 4 days. Color code corresponds to no (yellow), one (dark yellow), two (orange), three (red) and four or more (dark red) applications necessary to eradicate the infestation with the corresponding strategy.}
	\label{Fig:DependInitCond}
\end{figure}

\section{Discussion}
\label{sec:discussion}
Understanding the life cycle of head lice is an important step in knowing how to treat lice infestations, as the parasite behavior depends considerably on its age and gender. To this purpose we have proposed a mathematical model for a population structured by age and gender formulated as a system of PDEs~\eqref{model1_eq:PDEfemales}--\eqref{eq:initcond_pde}, which can be reduced to compartmental systems of delay differential equations~\eqref{model1_eq:DDE_all_inout} or ordinary differential equations~\eqref{ModelODE}. The latter was used to include treatments against head lice, which are differently eliminating eggs and nymphs/adult lice. To the best of our knowledge, besides the pioneer work by Laguna and coauthors~\cite{Laguna2011}, this is a quite unique study which proposes a mechanistic mathematical model for understanding the biology of the life cycle of head lice and attesting the efficiency of different treatments in eradicating lice infestations.\\ 
\ \\
Fundamental properties of the ODE model~\eqref{ModelODE} were studied in Section~\ref{sec:analysis}. Beside existence, uniqueness and nonnegativity of solutions we have considered existence and stability of equilibria of the dynamical system. Our results show that in case of a quarantined infected host, there might be either no lice (infection free equilibrium $P_0$) or an heterogeneous population with lice in all life stages. $P_0$ is locally asymptotically stable, hence small perturbations of this equilibrium might not lead to lice infection, even if untreated (cf. Fig.~\ref{Fig:DependInitCond}). Further, the analytical results suggest that there is no stationary state in which only juvenile or only adult lice sub-populations survive. Provided that the reproduction, maturation and survival parameters of the lice population satisfy $\RM>1$ and $\RW>1$, then the coexistence equilibrium $P_1$ exists, but it is unstable (Theorem~3) and if not treated, the lice population would grow uncontrolled (Fig.~\ref{Fig:No_treatment}). If the host is not isolated and lice transmission among infected hosts is possible, then there might exist two nontrivial equilibria $P_{2,3}$ (Theorem 4).\\
\ \\
By mean of computer experiments and numerical simulations we have studied (Sect.~\ref{sec:treat_simulax}) four possible treatments against head lice, namely topical non-ovocidal treatments (Strategy nr.~1), wet combing (Strategy nr.~2), combination of the two (Strategy nr. 3), and dimeticone-based products (Strategy nr. 4). For all products different efficacy and application schedules were studied. No  product  was assumed to be  100\%  successful in removing eggs, nymphs and adult lice, as this is technically not feasible~\cite{Speare2002,Feldmeier2012}. Of course, if such a product would exist re-treatment would not  be  required for isolated hosts.
If a (almost) perfectly working topical treatment which eliminates at least 90\% of live lice is available, then one application every 7 days repeated for three times is sufficient to eradicate moderate to severe infestations (Fig.~\ref{Fig6a_NT}). Relying on the biology of the lice life-cycle, the time gap between applications should not be too short or too long (Fig.~\ref{Fig:Strategy1}(c,d)), but could be relaxed to 14 days. For example, an effective strategy would be to apply a 100\% effective insecticide-based shampoo for three times, with two weeks breaks between one application and the next one (Fig.~\ref{Fig:Strategy1}e). In case of resistant lice, the duration of the treatment and the number of necessary applications increases (see e.g. for 40\% resistance, Fig.~\ref{Fig:shampoo60res} and Fig.~\ref{Fig6b_NT}).
Combing (Strategy nr. 2) is a useful method for detection, but according to our results, it could not be the method of choice to treat and eradicate a lice infestation.
Indeed, unless the lice population at detection is very small and combing is performed very carefully, removing at least half of the present eggs and live lice (Fig.~\ref{Fig6c_NT}), a high number of treatments could be necessary for the host to be lice free and the infestation could relapse within 2 weeks from the last treatment (Fig.~\ref{Fig:Strategy2}). Combining shampoos and combing (Strategy nr. 3) to treat a moderate to severe infestation could be quite time consuming and uncomfortable for the host due to the high number of applications required. If this method is chosen to treat an infestation, our results suggest to use effective products which can effectively remove eggs/live lice (Fig.~\ref{Fig:Strategy3}). Dimeticone-based products, in particular if a new application is repeated 4 days (rather than 8-10 days) after the previous one, allow for a lower number of applications even in case of severe infestations (Figs.~\ref{Fig:Strategy4} and~\ref{Fig:DependInitCond}). Our results indicate that early detection is crucial for quick and efficient eradication. Indeed, the number of treatment applications necessary to eradicate the infestation population increases with increasing eggs/live lice at the time of first treatment (Fig.~\ref{Fig:DependInitCond}).\\
\ \\ 
In Section~\ref{sec:treat_simulax} we have considered the case of a quarantined host. One might ask if treatments which have been shown to be effective for such hosts do also work when the host is not isolated. Let us consider a perfectly working topical treatment (as in Fig.~\ref{Fig:shampoo100}) and a host, say a pupil, who has been found infested with lice. We assume that upon detection the host is quarantined for one or two weeks and is treated with the perfectly working topical product once a week starting with the detection day (day 21), returning to school the day after the second (day 29, Fig.~\ref{Fig7a}) or the third treatment (day 36, Fig.~\ref{Fig7b}). Let us also assume that in the same classroom there is at least another host with undetected or not well treated lice infestation, so that upon returning to school, new lice could be transferred to our initial patient. Fig.~\ref{Fig7} shows that the treatment, which was effective for quarantined hosts (Fig.~\ref{Fig:shampoo100}), is failing for hosts who are at risk of reinfection. As long as the treated host goes back to an infectious environment, lice infestation cannot be eradicated. Based on our computational model results it would be advisable to reproduce a lice-free environment and minimize the reinfection risk, treating at the same time not only the first detected host but also his classmates. This is in agreement with the synchronized treatment strategy proposed in~\cite{Laguna2011,Meister2016}.\\
\ \\
Though the model was parametrized based on available literature on the biology of head lice, their life cycle and the estimated efficacy of different treatments against lice, for certain parameters (such as the transferring rates, see Table~\ref{Table_rates}) no data are available. Henceforth, the proposed mathematical model and resulting numerical simulations are not meant for data fitting but rather for understanding the time evolution of an infestation and predicting the performance of a possible treatment strategy. The sharp detection threshold which was used to assess the performance or determining the conclusion of a treatment could be put into question. The choice of a different value for the threshold would quantitatively modify the results presented here, as observed for the model proposed in~\cite{Laguna2011}. A further limitation of the proposed study lays in the deterministic nature of the model. The deterministic approach used here is appropriate for large populations (such as untreated lice colonies), whereas for populations with very few individuals a stochastic approach would be more suitable. A stochastic model could be adopted to improve the study of borderline cases such as those in Fig.~\ref{Fig6B} where, though the infestation could be considered eradicated, in the long run the lice population increases again. In a further study we plan to improve the modeling approach proposed here by considering a mixed approach of deterministic and stochastic processes, as it has been proposed in other fields of biology~\cite{Kraut2019}.

\begin{figure}[h!]
	\centering
		\begin{subfigure}{0.45\textwidth}
		\includegraphics[width=\textwidth]{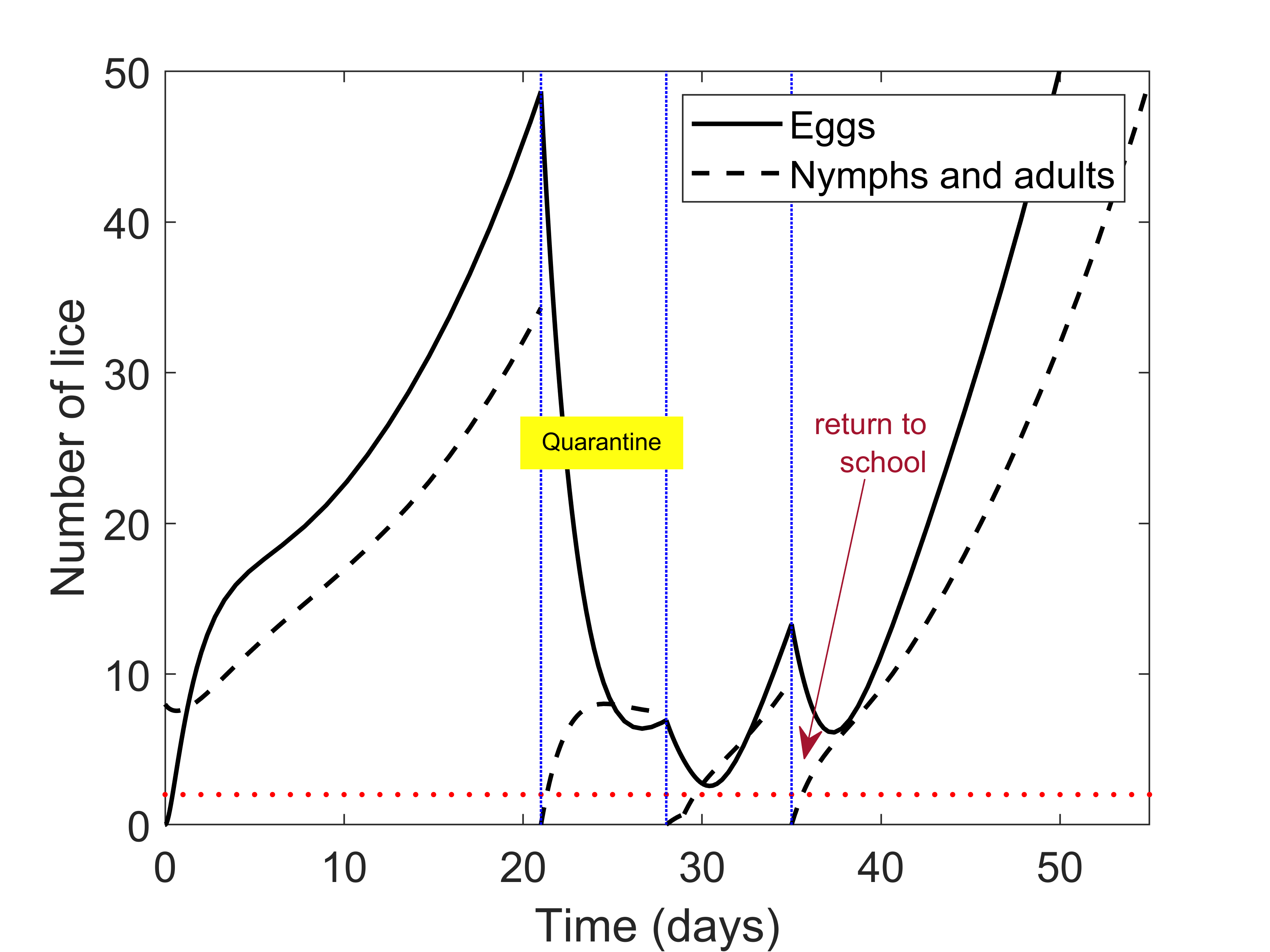}
		\caption{}
		\label{Fig7a}			
	\end{subfigure}
	\begin{subfigure}{0.45\textwidth}
		\includegraphics[width=\textwidth]{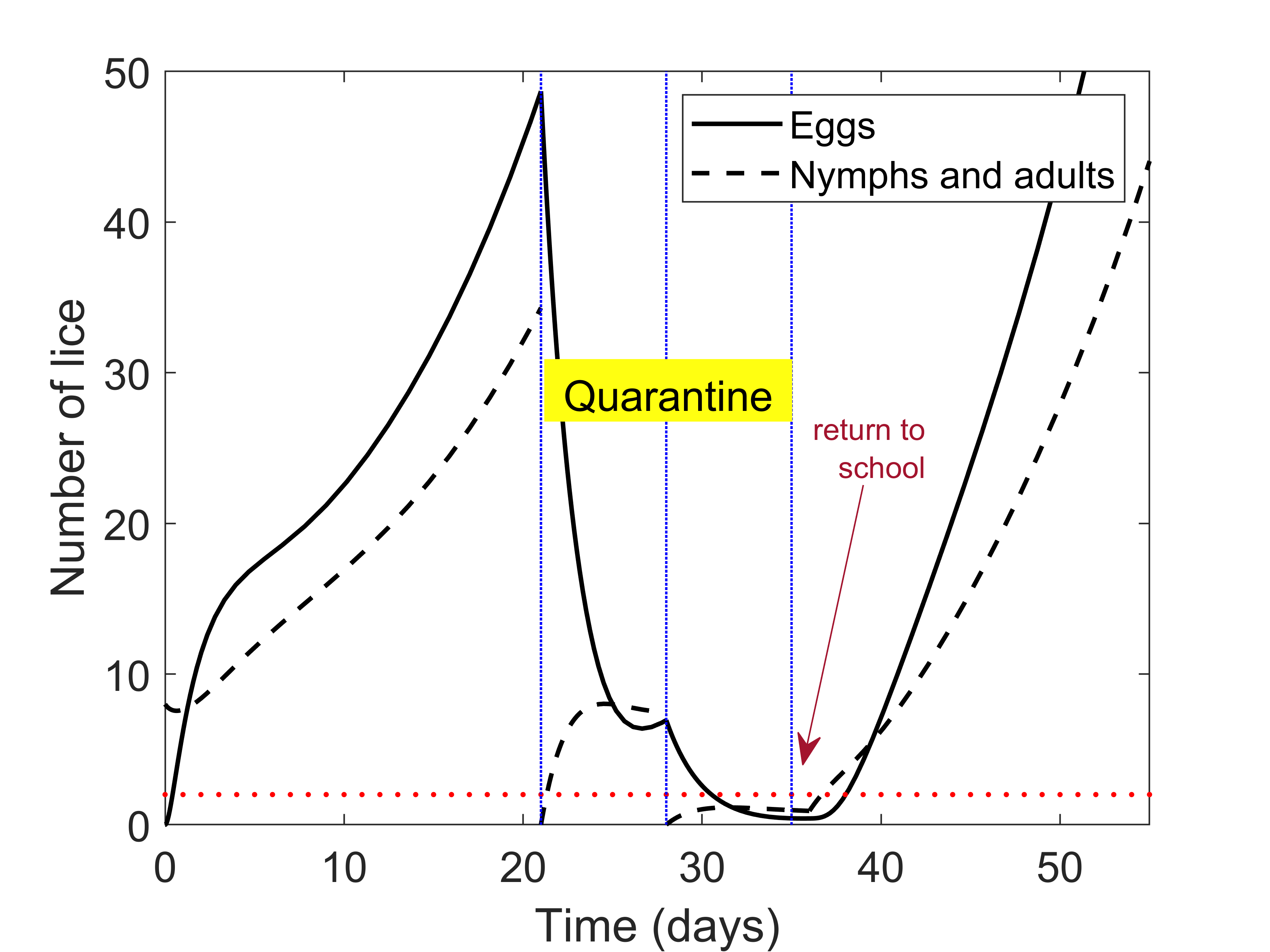}
		\caption{}
		\label{Fig7b}			
	\end{subfigure}
	\caption{Evolution in time of a lice colony which develops from a small group of adults in case of a non-isolated infestation. The host is treated with an ideally working shampoo (cf. Fig.~\ref{Fig:Strategy1}(b)) applied on day 21,28 and 35. Following detection (day 21) the lice are treated and a the host is quarantined for (a) one week, respectively, (b) two weeks. When the quarantine ends, the host returns to school, where others are infested and lice transmission is possible.}
	\label{Fig7}
\end{figure}

\section*{Acknowledgment}
MVB is supported by the European Social Fund and by the Ministry of Science, Research and Arts Baden-W\"urttemberg.

\begin{table}[b] 
	\begin{center}
		\begin{tabular}{c|l|c|c}
		    \hline\\[-0.8em]
			\textbf{Rate} & \textbf{Description} & \textbf{Value [Unit]} & \textbf{Reference}  \\
	    \hline\\[-0.8em]
	   $b_1$ & laid eggs per adult female & 3 [1/day] & \cite{Lebwohl2007,TakanoLee2003} \\
			$\mu_0$ & eggs death rate & 0.35 [1/day]& \cite{TakanoLee2003}\\
			$\mu_N$ & nymphs death rate & 0.195 [1/day]&  \cite{TakanoLee2003}\\
			$\mu_1$ & adult lice death rate & 1/30 [1/day]&  \cite{TakanoLee2003,Lebwohl2007}\\
			$\mu_B$ & breeding females death rate & 1/25 [1/day]& \cite{TakanoLee2003,Lebwohl2007} \\
			$1/\eta$ & egg stage duration & 7 [days] &  \cite{TakanoLee2003}\\
			$1/\omega$ & nymph stage duration & 9 [days] &  \cite{TakanoLee2003}\\
			$r$ & prob. for egg to turn into male louse & 0.367 & \cite{Perotti2004}\\
			$\rho$ & mating rate & 0.9 & \citep{TakanoLee2003,Lebwohl2007}\\
			$\theta$ & prob. for $\WB$ to return to $W$   & 100\% & \cite{Multiple_mating_Mehlhorn,Boutellis}\\
			$1/\alpha$ & duration of breeding stage & 3 [days] & \citep{TakanoLee2003} \\
			$\xi$ & death prob. during mating & 5\% & assumed \\
			$\beta_W$ & transferring rate females (outgoing) & 0.35 [1/day] & assumed \\
			$\beta_M$ & transferring rate males (outgoing) & 0.35 [1/day] & assumed\\
			$\alpha_W$ & transferring rate females (incoming) & {1} [1/day]& assumed \\
			$\alpha_M$ & transferring rate males (incoming) & {1} [1/day] & assumed\\
			\hline
		\end{tabular} 
		\caption{Model parameters, description and values used for numerical simulations.}
		\label{Table_rates}
	\end{center}
\end{table}

\bibliographystyle{abbrv}
\bibliography{NCMVB2019bib}
\end{document}